\documentclass[doublecolumn]{IEEEtran}
\usepackage{cite}

\usepackage{pgf} 
\usepackage{graphicx}

\usepackage{url}

\usepackage{latexsym,bm,amssymb,amsthm} 
\usepackage{color}
\usepackage{algorithm}
\usepackage{algpseudocode}

\usepackage{verbatim}
\usepackage{mathtools}
\mathtoolsset{showonlyrefs}

\usepackage{comment}

\usepackage{lipsum,cuted}

\newtheorem{thm}{Theorem}
\newtheorem{prop}{Proposition}
\newtheorem{rem}{Remark}
\newtheorem{defn}{Definition}
\newtheorem{exmp}{Example}
\newtheorem{assum}{Assumption}

\usepackage{balance}

\usepackage{tabularx}
\usepackage{makecell}

\begin{document}

\title{Deep Reinforcement Learning for Wireless Scheduling in Distributed Networked Control} 

\author{
	Gaoyang Pang, Kang Huang, Daniel E.\ Quevedo,~\IEEEmembership{Fellow,~IEEE}, \par 
	Branka Vucetic,~\IEEEmembership{Life Fellow,~IEEE}, Yonghui Li,~\IEEEmembership{Fellow,~IEEE}, Wanchun Liu,~\IEEEmembership{Senior Member,~IEEE}
	\vspace{-0.5cm}
    \thanks{G. Pang and K. Huang contributed equally to the work. The work of B. Vucetic was supported in part by the Australian Research Council Laureate Fellowship grant number FL160100032 and the ARC Discovery Project DP210103410. The work of Y. Li was supported by ARC under Grant DP250100462 and LP240100451. The work of W. Liu was supported by the Australian Research Council’s Discovery Early Career Researcher Award (DECRA) Project DE230100016. \textit{(Corresponding author: W. Liu.)}}
    \thanks{G. Pang, D. E. Quevedo, B. Vucetic, Y. Li, and W. Liu are with School of Electrical and Computer Engineering, The University of Sydney, Sydney, 2006, Australia (e-mail: \{gaoyang.pang, daniel.quevedo, branka.vucetic, yonghui.li, wanchun.liu\}@sydney.edu.au).}
    \thanks{K. Huang is with Huawei Shanghai Research Center, Shanghai, China, (e-mail: huangkang9@huawei.com).}
}
\maketitle

	
%

	
\begin{abstract}                 
		We consider a joint uplink and downlink scheduling problem of a fully distributed wireless networked control system (WNCS) with a limited number of frequency channels. Using elements of stochastic systems theory,  
		we derive a sufficient stability condition of the WNCS, which is stated in terms of both the control and communication system parameters.
		Once the condition is satisfied, there exists a stationary and deterministic scheduling policy that can stabilize all plants of the WNCS.
		By analyzing and representing the per-step cost function of the WNCS in terms of a finite-length countable vector state, we formulate the optimal transmission scheduling problem into a Markov decision process and develop a deep reinforcement learning (DRL) based framework for solving it.
        To tackle the challenges of a large action space in DRL, we propose novel action space reduction and action embedding methods for the DRL framework that can be applied to various algorithms, including Deep Q-Network (DQN), Deep Deterministic Policy Gradient (DDPG), and Twin Delayed Deep Deterministic Policy Gradient (TD3).
		Numerical results show that the proposed algorithm significantly outperforms benchmark policies. 
\end{abstract}
	
	\vspace{-1cm}
\section{Introduction}
The Fourth Industrial Revolution, Industry 4.0, is the automation of conventional manufacturing and
industrial processes through flexible mass production.
Automatic control in Industry 4.0 requires large-scale and interconnected deployment of massive spatially distributed industrial devices, such as sensors, actuators, machines, robots, and controllers. Eliminating the communication wires is a game-changer in reshaping traditional factories. In this regard, wireless networked control has become one of the most important technologies of Industry 4.0. It offers highly scalable and low-cost deployment capabilities and enables many industrial applications, such as smart cities, smart manufacturing, smart grids, e-commerce warehouses and industrial automation systems~\cite{ParkSurvey,KangJIoT}.

Unlike traditional cable-based networked control systems, in a large-scale wireless networked control system (WNCS), communication resources are limited.
Not surprisingly, during the past decade, state estimation and transmission scheduling in WNCS have drawn a lot of attention in the research community.
For example, some recent works have investigated the estimation problem under dynamic sampling periods~\cite{deng2023sequential} and random access protocol scheduling~\cite{liu2022distributed}.
Focusing on the sensor-controller communications only (the controller-actuator co-located scenario), the optimal transmission scheduling problem over a single frequency channel for achieving the best remote estimation quality was extensively investigated in~\cite{Tomlin,HAN2017260,leong2017sensor,DoubleThreshold}.
For the multi-frequency channel scenario, the optimal scheduling policy and structural results were obtained in~\cite{Wu2018Auto}.
The optimal transmission power scheduling problem of an energy-constrained remote estimation system was studied~in \cite{CaoEnergy21}.
The joint scheduling and power allocation problems of multi-plant-multi-frequency WNCSs were investigated in \cite{gatsis2015opportunistic,Eisen} for achieving the minimum overall transmission power consumption.
From a cyber-physical security perspective, the scheduling and control problems posed by denial‑of‑service (DoS) attackers were investigated in~\cite{WangICC22,JammerScheduleQin21,9346029}, for maximally deteriorating the WNCSs' performance.

In the WNCS examined in the above works, the scheduling problem is commonly reformulated into a Markov decision process (MDP) problem. The MDP problem can be solved using traditional algorithms (e.g., value iteration~\cite{Wu2018Auto,KangJIoT,KangTWC}, policy iteration~\cite{8642856}, and linear programming~\cite{chen2017delay,han2020joint}) and reinforcement learning algorithms (e.g., Q-learning~\cite{li2019multi,10944253} and deep Q-Networks~\cite{LEONG2020108759,DeepCAS}). Traditional algorithms and Q-learning are effective when the system scale is small.
For example, numerical results of the optimal scheduling of a two-sensor-one-frequency system using value iteration were presented in~\cite{Wu2018Auto}.
These methods involve operations over every possible state and action. As the number of actions increases, the number of state-action pairs grows, leading to a dramatic increase in the computations required. For example, when the action space is large, Q-learning necessitates constructing a Q-value table, which demands substantial storage proportional to the size of the action space, leading to the curse of dimensionality~\cite{li2019multi,zhu2017new,naeem2020generative}. Deep Q-Networks (DQN) address this issue using deep neural networks (DNNs) for value function approximations, thus eliminating the need for extensive Q-value tables.
Some recent works~\cite{LEONG2020108759,DeepCAS,10479170} have applied DQN to solve multi-system-multi-frequency scheduling problems in different WNCS scenarios.

However, handling large action spaces also presents significant challenges in DQN. The size of the DNNs required for DQN expands dramatically with respect to the action space. This can hinder effective exploration of the action space during training, resulting in increased storage and computational requirements. Ma et al.~\cite{ma2021hierarchical} proposed decoupling the original problem into subproblems and developing multiple deep reinforcement learning (DRL) algorithms to solve the corresponding subproblems with reduced decision space. However, the presence of multiple learning agents creates a non-stationary environment, which can destabilize learning and make convergence more difficult.
To deal with large decision spaces, recent works resort to stochastic policy-based DRL algorithms~\cite{ni2021multi,yang2019actor,hildebrandt2023opportunities}. A stochastic policy refers to a policy that specifies a probability distribution over actions, given the current state of the environment. Ni et al.~\cite{ni2021multi} leveraged a Graph Neural Network (GNN) based DRL to learn stochastic policies, i.e., Proximal Policy Optimization (PPO) for warehouse scheduling. Yang et al.~\cite{yang2019actor} developed an actor-critic DRL for learning stochastic policies with a continuous action space for scheduling, power allocation, and modulation scheme adaptation. However, it has been proved that the optimal policy is deterministic for unconstrained MDP problems~\cite{bertsekas2000dynamic}. Here, a deterministic policy means that the action in a given state is determined by a fixed function without any randomness. Consequently, deterministic policy-based DRL methods are generally preferable.

We note that most of the existing works focus on WNCSs which are partially distributed~\cite{Tomlin,HAN2017260,leong2017sensor,Wu2018Auto,gatsis2015opportunistic,Eisen,LEONG2020108759,DeepCAS}.
For example, Chen et al.~\cite{chen2023structure} focus on the sensor scheduling in the uplink transmission of a WNCS, using either DQN or Deep Deterministic Policy Gradient (DDPG). Redder et al.~\cite{redder2019deep} developed a DQN-based algorithm for optimal actuator scheduling in the downlink transmission of a WNCS over a Markov fading channel. Wang et al.~\cite{10632058} also developed a GNN-based multi-agent DRL for optimal actuator scheduling in a WNCS. Existing works consider only optimizing either the uplink or downlink transmission for WNCSs.
In a fully distributed setting, both downlink (controller-actuator) and uplink (sensor-controller) transmissions are crucial for stabilizing each plant.
The joint uplink and downlink scheduling problem of fully distributed multi-plant WNCS has never been considered in the open literature.

In this paper, we investigate the transmission scheduling problem of distributed WNCS. The main contributions are summarized as follows:
\begin{itemize}
	\item We propose a distributed $N$-plant-$M$-frequency WNCS model, where the controller schedules the uplink and downlink transmissions of all $N$ plants and the spatial diversity of different communication links is taken into account. The controller generates sequential predictive control commands for each of the plants based on pre-designed deadbeat control laws.\footnote{Note that deadbeat controller is commonly considered as a time-optimal controller that takes the minimum time for setting the current plant state to the origin.}
	Different from uplink or downlink only scheduling, a joint scheduling algorithm has a larger action space and also needs to automatically balance the trade-off between the uplink and the downlink due to the communication resource limit.
	To the best of our knowledge, joint uplink and downlink transmission scheduling of distributed WNCSs has not been investigated in the open literature.

	\item We derive a sufficient stability condition of the WNCS in terms of both the control and communication system parameters.
	The result provides a theoretical guarantee that there exists at least one stationary and deterministic scheduling policy that can stabilize all plants of the WNCS.
	We show that the obtained condition is also necessary in the absence of spatial diversity of different communication links.
	
	\item We construct a finite-length countable vector state of the WNCS in terms of the time duration between consecutive received packets at the controller and the actuators.
	Then, we prove that the per-step cost function of the WNCS is determined by the time-duration-related vector state.
	Building on this, we formulate the optimal transmission scheduling problem into an MDP problem with a countable state space for achieving the minimum expected total discounted cost.
	We propose effective action space reduction and action embedding methods that can be applied to DRL algorithms, including foundational ones like DQN and more advanced ones such as DDPG and Twin Delayed Deep Deterministic Policy Gradient (TD3), for solving the problem.	
	Numerical results illustrate that the proposed algorithm can reduce the expected cost significantly compared to available benchmark policies.

\end{itemize}

Notations: $\sum_{m=i}^{j} a_m \triangleq 0$ if $i>j$.
$\limsup_{K \rightarrow\infty} $ is the  limit superior operator.
$\mathsf{C}^k_n \triangleq \frac{n!}{k!(n-k)!}$ and $\mathsf{P}^k_n\triangleq \frac{n!}{(n-k)!}$  are the numbers of combinations and permutations of $n$ things taken $k$ at a time, respectively. $\text{Tr}(\mathbf{A})$ and $\mathsf{rank}[\mathbf{A}]$ denote the trace and the rank of matrix $\mathbf{A}$, respectively.

\section{Distributed WNCS with Shared Wireless Resource}
We consider a distributed WNCS system with $N$ plants and a central controller as illustrated in Fig.~\ref{fig:sys}. Each smart sensor employs a local Kalman filter. The output of plant $i$ is measured by smart sensor $i$, which sends pre-filtered measurements (local state estimates) to the controller.
The central controller applies a remote (state) estimator and a control algorithm for plant $i$. It then generates and sends a control signal to the actuator $i$, thereby closing the loop.
The uplink (sensor-controller) and downlink (controller-actuator) communications for the $N$ plants share a common wireless network with only $M$ frequency channels, where $M<2N$. Thus, not every node is allowed to transmit at the same time and communications need to be scheduled. As shown in Fig.~\ref{fig:sys}, we shall focus on a setup where scheduling is done at the controller side, which schedules both the downlink and uplink transmissions.\footnote{Such a setup is practical, noting that most of the existing wireless communication systems, including 5G, support both uplink and downlink at base stations~\cite{3gpp.38.300}.}

\begin{figure}[t]
	\centering\includegraphics[scale=0.45]{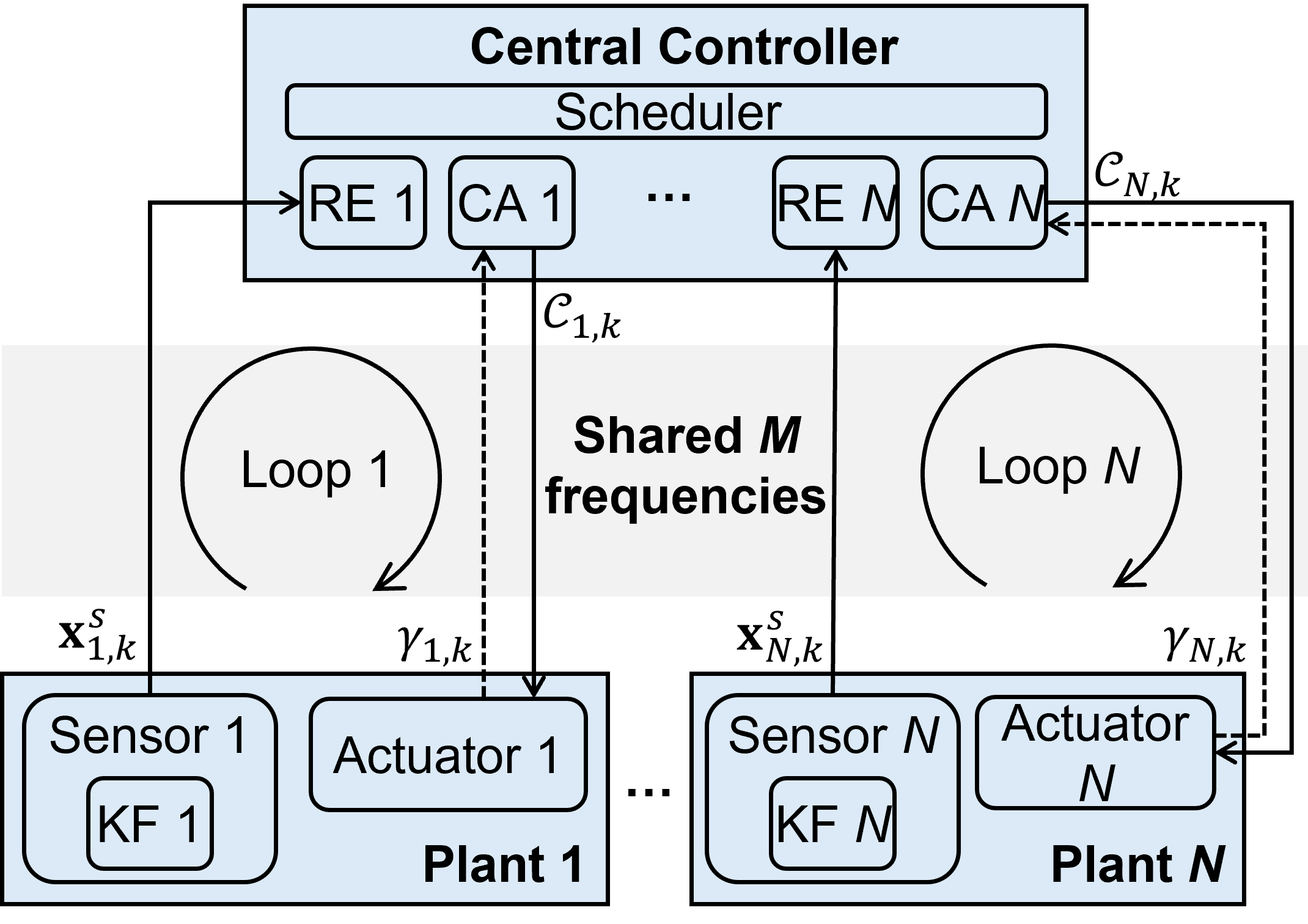}
	\vspace{-0.3cm}	
	\caption{A distributed networked control system with $N$ plants sharing $M$ frequency channels. Kalman filter, remote estimator and control algorithm are denoted as KF, RE and CA and discussed in Sections~\ref{sec:KF}, \ref{sec:RE} and \ref{sec:control}, respectively.}
	\label{fig:sys}
\end{figure}

\subsection{Plant Dynamics}
The $N$ plants are modeled as linear time-invariant (LTI) discrete-time systems as \cite{gatsis2015opportunistic,DeepCAS}
\begin{equation} \label{plant_state}
	\begin{aligned}
	&\mathbf{x}_{i,k+1} = \mathbf{A}_i\mathbf{x}_{i,k} + \mathbf{B}_i\mathbf{u}_{i,k} + \mathbf{w}_{i,k},\\
	&\mathbf{y}_{i,k} = \mathbf{C}_i\mathbf{x}_{i,k} + \mathbf{v}_{i,k}, \quad i = 1,\cdots, N
	\end{aligned}
\end{equation} 
where $\mathbf{x}_{i,k} \in \mathbb{R}^{n_i}$  and $\mathbf{u}_{i,k} \in \mathbb{R}^{m_i}$ are the state vector of plant $i$ and the control input applied by actuator $i$ at time $k$, respectively.
$\mathbf{w}_{i,k} \in \mathbb{R}^{n_i}$ is the $i$-th plant disturbance and is an independent and  identically distributed (i.i.d.) zero-mean Gaussian white noise process with   covariance matrix $\mathbf{Q}_i^w \in \mathbb{R}^{n_i \times n_i}$.  $\mathbf{A}_i \in \mathbb{R}^{n_i \times n_i}$ and $\mathbf{B}_i \in \mathbb{R}^{n_i \times m_i}$ are the system-transition matrix and control-input matrix for plant $i$, respectively. 
$\mathbf{y}_{i,k} \in \mathbb{R}^{p_i}$ is sensor $i$'s measurement  of plant $i$ at time $k$ and $\mathbf{v}_{i,k} \in \mathbb{R}^{p_i}$ is the measurement noise, modeled as an i.i.d.\ zero-mean Gaussian white noise process with covariance matrix $\mathbf{Q}_i^v \in \mathbb{R}^{p_i \times p_i}$. $\mathbf{C}_i \in \mathbb{R}^{p_i \times n_i}$ is the measurement matrix of plant $i$. 
We assume that plant $i$ is $v_i$-step controllable, $\forall i \in \{1,\dots,N\}$~\cite{Burak}, i.e., there exists a control gain $\tilde{\mathbf{K}}_i$ satisfying
\begin{equation}\label{control_property}
(\mathbf{A}_i+\mathbf{B}_i\tilde{\mathbf{K}}_i)^{v_i} = \mathbf{0}.
\end{equation}
Note that when $\mathbf{A}_i$ is
non-singular, the system~\eqref{plant_state} is $v_i$-step controllable if and only if
\begin{equation}
\mathsf{rank}\left[\mathbf{A}^{-1}_i \mathbf{B}_i, \mathbf{A}^{-2}_i \mathbf{B}_i, \dots,\mathbf{A}^{-v_i}_i \mathbf{B}_i\right] = n_i.
\end{equation}
The details of the control algorithm will be given in Section~\ref{sec:control}.

\subsection{Smart Sensors}\label{sec:KF}
Due to the measurement noise, each smart sensor runs a Kalman filter to estimate the current plant state as below~\cite{liu2020remote}:\footnote{In this subsection, we focus on smart sensor $i$ and the index $i$ of each quantity is omitted for clarity.}

\begin{equation}\label{eq:KF}
	\begin{aligned} 
	\mathbf{x}^s_{k|k-1} &= \mathbf{A}\mathbf{x}^s_{k-1} + \mathbf{B}\mathbf{u}_{k-1} \\
	\mathbf{P}^s_{k|k-1} &= \mathbf{A}\mathbf{P}^s_{k-1}\mathbf{A}^{\top} + \mathbf{Q}^w \\
	\mathbf{K}_k &= \mathbf{P}^s_{k|k-1}\mathbf{C}^{\top}(\mathbf{C}\mathbf{P}^s_{k|k-1}\mathbf{C}^{\top} + \mathbf{Q}^v)^{-1} \\
	\mathbf{x}^s_{k} &= \mathbf{x}^s_{k|k-1} + \mathbf{K}_k(\mathbf{y}_k - \mathbf{C}\mathbf{x}^s_{k|k-1}) \\
	\mathbf{P}^s_{k} &= (\mathbf{I}-\mathbf{K}_k\mathbf{C})\mathbf{P}^s_{k|k-1} 
	\end{aligned}
\end{equation}
where $\mathbf{x}^s_{k|k-1}$ and $\mathbf{x}^s_{k}$ are the prior and posterior state estimation at time $k$, respectively, and  $\mathbf{P}^s_{k|k-1}$ and $\mathbf{P}^s_{k}$ are the (estimation error) covariances of $\mathbf{e}^s_{k|k-1} \triangleq \mathbf{x}_{k} - \mathbf{x}^s_{k|k-1}$ and $\mathbf{e}^s_{k} \triangleq \mathbf{x}_k - \mathbf{x}^s_{k}$, respectively. 
$\mathbf{K}_k$ is the Kalman gain at time $k$.
We assume that  each $(\mathbf{A},\mathbf{C})$ is observable and $(\mathbf{A},\mathbf{Q}^w)$ is controllable~\cite{Wu2018Auto}. Thus, the Kalman gain $\mathbf{K}_k$ and error covariance matrix $\mathbf{P}^s_{k}$ converge to constant matrices $\hat{\mathbf{K}}$ and $\hat{\mathbf{P}}^s$, respectively, i.e., the smart sensor is in the stationary mode.
As shown in Fig.~\ref{fig:sys}, the smart sensor is co-located with the actuator in the plant~\cite{mishra2020stochastic} and can obtain the control input in~\eqref{eq:KF}.
Thus, the Kalman filter \eqref{eq:KF} is the optimal estimator of the linear system \eqref{plant_state} in terms of the estimation mean-square error~\cite{kailath1980linear}.

As shown in Fig.\ \ref{fig:sys}, at every scheduling instant $k$, the smart sensor sends the local estimate ${\mathbf{x}}_k^s$ (rather than the raw measurement $\mathbf{y}_k$) to the controller~\cite{liu2020remote}.

Before proceeding, we note that~\eqref{eq:KF} leads to the following recursion for the  estimation error at the sensors:
\begin{equation}
\mathbf{e}^s_{k+1|k} = \mathbf{A}\mathbf{e}^s_{k} + \mathbf{w}_k
\end{equation}
\begin{equation}
\mathbf{e}^s_{k} = \big(\mathbf{I}-\hat{\mathbf{K}}\mathbf{C}\big)\mathbf{e}^s_{k|k-1} - \hat{\mathbf{K}}\mathbf{v}_k
\end{equation}
and hence
\begin{equation}
\begin{aligned}
\mathbf{e}^s_k 
= \big(\mathbf{I}-\hat{\mathbf{K}}\mathbf{C}\big)\mathbf{A}\mathbf{e}^s_{k-1} + \big(\mathbf{I}-\hat{\mathbf{K}}\mathbf{C}\big)\mathbf{w}_{k-1} - \hat{\mathbf{K}}\mathbf{v}_k.
\end{aligned}
\end{equation}
Then, the relation between the estimation errors $\mathbf{e}_{k}^s$ and $\mathbf{e}_{k-K}^s$, $\forall K \in \mathbb{N}$, is established as
\begin{equation} \label{error_sensor}
\begin{aligned}
\mathbf{e}^s_k 
&\!=\! \mathbf{Z}^K\mathbf{e}^s_{k-\!K} \!+\!\! \sum_{i=1}^{K}\mathbf{Z}^{i\!-\!1}\big(\mathbf{I}\!-\!\hat{\mathbf{K}}\mathbf{C}\big)\mathbf{w}_{k-\!i} 
 \!-\!\! \sum_{i=1}^{K}\mathbf{Z}^{i-\!1}\hat{\mathbf{K}}\mathbf{v}_{k\!-i+1},
\end{aligned}
\end{equation}
where $\mathbf{Z} \triangleq \big(\mathbf{I}-\hat{\mathbf{K}}\mathbf{C}\big)\mathbf{A}$.

\subsection{Remote Estimation}\label{sec:RE}
At the beginning of a time slot, the controller calculates the control sequence and transmits it within the time slot based on the current plant state estimation rather than the real-time state, as the packet carrying the current plant state can only be received at the end of the current time slot due to the one-step transmission delay. Then, the controller applies a minimum mean-square error (MMSE) remote estimator for each plant taking into account the random packet dropouts and one-step transmission delay as~\cite{liu2020remote}:
\begin{equation}\label{eq:RE}
\hat{\mathbf{x}}_{i,k+1} = \begin{cases}
\mathbf{A}_i\mathbf{x}^s_{i,k} + \mathbf{B}_i\mathbf{u}_{i,k} &\mbox{ if } \beta_{i,k} = 1\\
\mathbf{A}_i\hat{\mathbf{x}}_{k} + \mathbf{B}_i\mathbf{u}_{i,k} &\mbox{ if } \beta_{i,k} = 0\\
\end{cases}
\end{equation}
where $\beta_{i,k} = 1$ or $0$ indicates that the controller receives sensor $i$'s packet or not at time $k$, respectively.
Then, the estimation error $\mathbf{e}_{i,k}$ is obtained as 
\begin{equation} \label{error_update}
\mathbf{e}_{i,k} \triangleq \mathbf{x}_{i,k} - \hat{\mathbf{x}}_{i,k} = \begin{cases}
\mathbf{A}_i\mathbf{e}^s_{i,k-1} + \mathbf{w}_{i,k-1} &\mbox{ if } \beta_{i,k-1} = 1\\
\mathbf{A}_i\mathbf{e}_{i,k-1} + \mathbf{w}_{i,k-1} &\mbox{ if } \beta_{i,k-1}=0
\end{cases}
\end{equation}

From~\eqref{eq:RE}, the controller's current estimation depends on the most recently received sensor estimation. Let $\tau_{i,k} \in  \{1,2,\cdots\}$ denote the \emph{age-of-information (AoI)}~(see \cite{KangTWC} and reference therein) of sensor $i$'s packet observed at time $k$, i.e., the number of elapsed time slots since the latest successfully delivered sensor $i$'s packet before the current time $k$, which reflects how old
the most recently received sensor measurement is.
Then, it is easy to see that the updating rule of $\tau_{i,k}$ is given by
\begin{equation} \label{tau_update}
\tau_{i,k+1} = \begin{cases}
1 &\mbox{ if } \beta_{i,k} = 1\\
\tau_{i,k} + 1 &\mbox{ otherwise.}
\end{cases}
\end{equation}
Using \eqref{error_update} and the AoI, the relation between the local and remote estimation error   can be characterized by:
\begin{equation} \label{error_local_remote}
\mathbf{e}_{i,k} = \mathbf{A}_{i}^{\tau_{i,k}}\mathbf{e}^s_{i,k-\tau_{i,k}} + \sum_{j = 1}^{\tau_{i,k}}\mathbf{A}_i^{j-1}\mathbf{w}_{k-j}.
\end{equation}

\subsection{Control Algorithm}\label{sec:control}
Due to the fact that downlink transmissions are unreliable, actuator $i$ may not  receive the controller’s control-command-carrying packets, even when transmissions are scheduled.
We adopt a predictive control approach~\cite{Burak,liu2020anytime} to provide robustness against packet failures: the controller sends a length-$v_i$ sequence of control commands including both the current command and the predicted future commands to the actuator once scheduled; if the current packet is lost, the actuator applies the previously received predictive command as the control input for the current time slot.

The control sequence for plant $i$ is generated by a linear deadbeat control gain $\tilde{\mathbf{K}}_i$ as~\cite{Burak}
\begin{equation} \label{predictive_control}
\mathcal{C}_{i,k} = \big[\tilde{\mathbf{K}}_{i}\hat{\mathbf{x}}_{i,k}, \tilde{\mathbf{K}}_{i}\mathbf{\Phi}_i\hat{\mathbf{x}}_{i,k},\cdots,\tilde{\mathbf{K}}_i(\mathbf{\Phi}_i)^{v_i-1}\hat{\mathbf{x}}_{i,k}\big]
\end{equation} 
where $\mathbf{\Phi}_i\triangleq \mathbf{A}_i+\mathbf{B}_i\tilde{\mathbf{K}}_i$ and $\tilde{\mathbf{K}}_i$ satisfies
$
(\mathbf{A}_i+\mathbf{B}_i\tilde{\mathbf{K}}_i)^{v_i} = \mathbf{0},
$
and $v_i$ is the controllability index of the pair $(\mathbf{A}_i,\mathbf{B}_i)$.
Note that the first element in $\mathcal{C}_{i,k}$ is the current control command and the rest are the predicted ones.

\begin{rem}
It can be verified that if the current state estimation  $\hat{\mathbf{x}}_{i,k}$ is perfect and the plant $i$ is disturbance free, then the plant state $\mathbf{x}_{i,k}$ would be set to zero after applying all $v_i$ steps of the control sequence in \eqref{predictive_control}.
Such a deadbeat controller is commonly considered as a time-optimal controller that takes the minimum time for setting the current plant state to the origin~\cite{deadbeat}.
We note that the deadbeat control law may not be cost-optimal to minimize the quadratic cost function defined in Section~\ref{sec:mainMDP}, and the optimal control law may depend on the scheduling policy. Since the current work focuses on the transmission scheduling design of the $N$-plant-$M$-frequency WNCS, the optimal joint control-scheduling problem can be investigated in our future work. Specifically, we aim to identify an optimal scheduling policy that minimizes the quadratic cost for a given deadbeat controller.
\end{rem}

Accordingly, actuator $i$ maintains a length-$v_i$ buffer
\begin{equation}
\mathcal{U}_{i,k} \triangleq [\mathbf{u}_{i,k}^0, \mathbf{u}_{i,k}^1,\cdots,\mathbf{u}_{i,k}^{v_i-1}]
\end{equation}
to store the received control commands.
If the current control packet is received, the buffer is reset with received sequence; otherwise, it is shifted one step forward as
\begin{equation}\label{eq:U_seq}
\mathcal{U}_{i,k} = \begin{cases}
\mathcal{C}_{i,k}, &\mbox{ if } \gamma_{i,k} = 1\\
[\mathbf{u}_{i,k-1}^1, \mathbf{u}_{i,k-1}^2,\cdots,\mathbf{u}_{i,k-1}^{v_i},\mathbf{0}], &\mbox{ if } \gamma_{i,k} = 0\\
\end{cases}
\end{equation}
where $\gamma_{i,k} = 1$ or $0$ indicate that actuator $i$ receives a control packet or not at time $k$, respectively.
The first command in the buffer is applied as the control input each time
\begin{equation}
\mathbf{u}_{i,k} \triangleq \mathbf{u}_{i,k}^0.
\end{equation}

Let $\eta_{i,k} \in \{1,2,\cdots\}$ denote the AoI of the controller's packet at actuator $i$ observed at time $k$, i.e., the number of elapsed time slots (including the current time slot) since the actuator $i$ last received a control packet. Its  updating rule is given as
\begin{equation} \label{eta_update}
\eta_{i,k} = \begin{cases}
1, &\mbox{ if } \gamma_{i,k} = 1\\
\eta_{i,k-1} + 1 &\mbox{ if } \gamma_{i,k} = 0
\end{cases}
\end{equation}
From~\eqref{predictive_control} and \eqref{eq:U_seq}, and by using the deadbeat control property~\eqref{control_property}, the applied control input can be concisely written as
\begin{equation} \label{apply_u}
\mathbf{u}_{i,k} = \tilde{\mathbf{K}}_i(\mathbf{\Phi}_i)^{\eta_{i,k}-1}\hat{\mathbf{x}}_{i,k+1-\eta_{i,k}}.
\end{equation}

\subsection{Communication Scheduler}\label{sec:schedule}
The $N$-plant WNCS has $N$ uplinks and $N$ downlinks sharing $M$ frequencies. Each frequency can be occupied by at most one link, and each link can be allocated to at most one frequency at a time.
Let $a_{m,k} \in \{-N,\dots, 0, \dots, N\}$ denote the allocated link to frequency $m$ at time $k$, where $a_{m,k} = i'$ means the frequency is allocated to $|i'|$-th plant, where $i'>0$ and $<0$ indicate for uplink and downlink, respectively, and $i'=0$ denotes that the frequency channel is idle.

The packet transmissions of each link are modeled as i.i.d.\ packet dropout processes.
Unlike most of the existing works wherein transmission scheduling of WNCSs assumes that transmissions from different nodes on the same frequency channel have the same packet drop probability \cite{LEONG2020108759}, we here consider a more practical scenario by taking into account the spatial diversity of different transmission nodes -- each frequency has different dropout probabilities for different uplink and downlink transmissions.
The packet success probabilities of the uplink and downlink of plant $i$ at frequency $m$ are given by $\xi^s_{m,i}$ and $\xi^c_{m,i}$, respectively,
where
\begin{equation}\label{eq:prob}
\begin{aligned}
&\xi^s_{m,i} \triangleq \mathbb{P}[\beta_{i,k} = 1|a_{m,k} = i],\\
&\xi^c_{m,i} \triangleq \mathbb{P}[\gamma_{i,k} = 1|a_{m,k} = -i].
\end{aligned}
\end{equation}

The packet success probabilities can be estimated  by the controller  based on standard channel estimation techniques~\cite{tse2005fundamentals,CSI}, and are utilized by the MDP and DRL-based solutions in Section~\ref{sec:MDP} and Section~\ref{sec:RL}. 

\emph{Acknowledgment feedback.}
We note that the Transmission Control Protocol (TCP) is commonly adopted in commercial telecommunication systems for data transmission, such as 5G and WiFi~\cite{Schenato07ProcIEEE,7164323,8039300}. TCP relies on acknowledgment feedback, where the receiver sends a one-bit acknowledgment signal back to the sender to confirm the successful receipt of data packets. This feedback mechanism ensures reliable data transmission by verifying the delivery of data. Therefore, we assume that both the uplink and downlink transmissions adopt the acknowledgment feedback scheme. In particular, the actuator sends a one-bit feedback signal of $\gamma_{i,k}$ to the controller, and the controller sends  $\beta_{i,k}$ as well as the received $\gamma_{i,k}$ to sensor $i$ each time with negligible overhead.
From $\{\beta_{i,k}\}$ and \eqref{eq:RE}, the smart sensor knows the estimated plant state $\hat{\mathbf{x}}_{i,k}$ by the controller; then, using $\hat{\mathbf{x}}_{i,k}$, $\{\gamma_{i,k}\}$ and \eqref{apply_u}, the sensor can calculate the applied control input $\mathbf{u}_{i,k}$, which is utilized for local state estimation as mentioned in \eqref{eq:KF}.
Note that sending one-bit acknowledgment feedback leads to negligible overhead in contrast to sending the applied control input $\mathbf{u}_{i,k}$ from the controller to the sensors.

\section{Stability Condition}\label{sec:stability}
Before turning to designing scheduling policies, it is critical to elucidate conditions that the WNCS needs to satisfy to ensure that there exists at least one stationary and deterministic scheduling policy that can stabilize all plants using the available network resources. 
Note that a policy $\pi$ is a function mapping from a state to an action. A deterministic policy means that the function gives the same action when the input state is fixed.
A \textbf{stationary policy} means the function $\pi$ is time invariant, i.e., $\pi_k=\pi,\forall k$. In other words, the action only depends on the state, not the time~\cite{bertsekas2000dynamic}.
We adopt a very commonly considered stochastic stability condition as below.
\setcounter{thm}{0}
\begin{assum} \label{ass:initial}
	The expected initial quadratic norm of each plant state is bounded, i.e.,
	$\mathbb{E}[\mathbf{x}_{i,0}^\top\mathbf{x}_{i,0}]<\infty, \forall i = 1,\dots,N$.	
\end{assum}	

\setcounter{thm}{0}
\begin{defn}[Mean-Square Stability] \label{def:stability}
	The WNCS is mean-square stable under Assumption~\ref{ass:initial}  if and only if 
	\begin{equation}\label{eq:stability_def}
	\limsup_{K \rightarrow\infty} \frac{1}{K} \sum_{k=1}^{K}
	\mathbb{E}[\mathbf{x}_{i,k}^\top\mathbf{x}_{i,k}] < \infty, \forall i=1,\dots,N.
	\end{equation}	
	
\end{defn}

Definition~\ref{def:stability} establishes mean‐square stability as the condition that the long‐term average of the expected squared state remains bounded. Our stability analysis is based on a stochastic framework that emphasizes average performance over time rather than focusing on instantaneous dropout events. In practice, there may be brief periods when $\gamma_{i,k}=0$, that is, when an actuator consecutively fails to receive new control commands. Although such short bursts can temporarily degrade performance, the overall system remains stable as long as these events occur infrequently enough. In other words, even if consecutive packet dropouts occur occasionally, the average packet success rate will ensure that the state does not grow unbounded. Thus, the stability condition guarantees that, over the long run, the system remains mean‐square stable under a properly designed scheduling protocol, rather than requiring stability at every single time instant. In the later analysis, we provide a specific scheduling policy with stability guarantees.

Intuitively, the stability condition of the WNCS should depend on both the (open-loop) unstable plant systems (i.e., those where $\rho(\mathbf{A}_i)\geq 1$) and the $M$-frequency communication system parameters. 
For the tractability of sufficient stability condition analysis (i.e., to prove the existence of a stabilizing policy), we focus on a policy class that groups the unstable plants into $M$ disjoint sets $\mathcal{F}_1,\dots,\mathcal{F}_M$ and allocates them to the $M$ frequencies, accordingly.
Let $\bar{\mathcal{F}}\triangleq \{i:\rho(\mathbf{A}_i)\geq 1,i\in\{1,\dots,N\}\}$ denote the index set of all unstable plants. We have
$\mathcal{F}_1\cup\dots\cup\mathcal{F}_M = \bar{\mathcal{F}} \subseteq \{1,\dots,N\}$  and
$\mathcal{F}_i\cap\mathcal{F}_j =\emptyset, \forall i\neq j$.
Then, we present the stability condition below which takes into account all potential allocations $(\mathcal{F}_1,\dots,\mathcal{F}_M)$.
\setcounter{thm}{1}

\setcounter{thm}{0}
\begin{thm}[Stabilizability]\label{theory:stability}
Consider the index set $\{\mathcal{F}_m\}$ as introduced above and define
$\rho^{\max}_{m} \triangleq \max_{i\in \mathcal{F}_m}\rho^2(\mathbf{A}_i)$, $\bar{\xi}^{\max}_{m} \triangleq \max_{i\in \mathcal{F}_m}\{\bar{\xi}^s_{m,i},\bar{\xi}^c_{m,i}\}$, $\bar{\xi}^s_{m,i} \triangleq 1 -{\xi}^s_{m,i}$, $\bar{\xi}^c_{m,i} \triangleq 1 -{\xi}^c_{m,i}$. We then have:
\par	(a) A sufficient condition under which the WNCS described by \eqref{plant_state}, \eqref{eq:KF}, \eqref{eq:RE}, \eqref{apply_u} and \eqref{eq:prob} has a stationary and deterministic scheduling policy satisfying the stability condition~\eqref{eq:stability_def} is given by
	\begin{equation}\label{eq:stability}
	\kappa \triangleq \min_{(\mathcal{F}_1,\dots,\mathcal{F}_M)}
	\max_{m=1,\dots,M,\mathcal{F}_m\neq \emptyset} \rho^{\max}_{m}\bar{\xi}^{\max}_{m}<1,
	\end{equation}
    where the operation $\min$ is taken over the set of all possible partitions $(\mathcal{F}_1,\dots,\mathcal{F}_M)$ of the set of unstable plants. For each such partition, we compute the maximum value over the $M$ frequencies, and then the minimum over all these partitions gives the most favorable (i.e., the smallest) worst-case product across the frequencies. The terms $\rho^{\max}_{m}$ and $\bar{\xi}^{\max}_{m}$ are computed with respect to the subset $\mathcal{F}_m$ of unstable plants.
    \par 
	(b) For the special case that the packet error probabilities of different links are identical at the same frequency (i.e., where no  spatial diversity exists), $\bar{\xi}^s_{m,1}=\bar{\xi}^c_{m,1} = \dots = \bar{\xi}^s_{m,N}=\bar{\xi}^c_{m,N}, \forall m=1,\dots,M$,
	the condition \eqref{eq:stability} is necessary and sufficient.

\end{thm}
\begin{proof}
The sufficient condition \eqref{eq:stability} is derived in two steps: 1) the construction of a stationary and deterministic scheduling policy and 2) the proof of the condition under which the constructed policy leads to a bounded average cost.
	Since a plant with $\rho(\mathbf{A}_i)<1$ does not need any communication resources for stabilization, in the following, we only focus on the unstable plants in $\bar{\mathcal{F}}\neq \emptyset$.
	
	We construct a multi-frequency persistent scheduling policy: the unstable plants are grouped into $M$ sets $\mathcal{F}_1,\dots,\mathcal{F}_M$, corresponding to the $M$ frequencies.
	In each frequency, the controller schedules the uplink of the first plant, say plant $i$, persistently until success. It then persistently schedules the downlink until success, and waits for $v_i-1$ steps for applying all the control commands in the actuator's buffer. It then schedules the uplink of the second plant, and so forth. This procedure is repeated ad-infinitum.
	We choose such a policy because of its tractability and the tightness of the sufficient stability condition that we will derive.
	The detailed proof is included in~\cite{liu2021DRL}.
\end{proof}

Theorem~\ref{theory:stability} establishes that, once condition \eqref{eq:stability} is satisfied, there exists at least one stationary and deterministic scheduling protocol that guarantees mean-square stability of all plants. This does not imply that any arbitrary protocol will stabilize the system; rather, condition \eqref{eq:stability} ensures the existence of a suitable policy—such as the multi-frequency persistent scheduling approach we construct in the proof.


\setcounter{thm}{0}
\begin{exmp}
	Consider a WNCS with $N=3$ and $M=2$, and $\rho^2(\mathbf{A}_i)\geq 1, i=1,2,3$. The packet error probabilities are $\bar{\xi}^s_{1,1}=0.1$, $\bar{\xi}^c_{1,1}=0.3$, $\bar{\xi}^s_{1,2}=0.2$, $\bar{\xi}^c_{1,2}=0.1$, $\bar{\xi}^s_{1,3}=0.2$, $\bar{\xi}^c_{1,3}=0.4$, 
	$\bar{\xi}^s_{2,1}=0.1$, $\bar{\xi}^c_{2,1}=0.3$, $\bar{\xi}^s_{2,2}=0.2$, $\bar{\xi}^c_{2,2}=0.1$, $\bar{\xi}^s_{2,3}=0.2$, $\bar{\xi}^c_{2,3}=0.4$.
	There are eight plant-grouping schemes at the two frequencies $(\mathcal{F}_1,\mathcal{F}_2)$, i.e., $(\{1,2,3\},\emptyset)$, $(\{1,2\},\{3\})$, $(\{1\},\{2,3\})$, $(\{2,3\},\{1\})$, $(\{2\},\{1,3\})$, $(\{3\},\{1,2\})$, $(\{1,3\},\{2\})$ and $(\emptyset, \{1,2,3\})$.
	If $\rho^2(\mathbf{A}_1) =3$, $\rho^2(\mathbf{A}_2) =2$ and $\rho^2(\mathbf{A}_3) =1$, the stability condition is satisfied as $\kappa= 0.9<1$; if $\rho^2(\mathbf{A}_1) =1$, $\rho^2(\mathbf{A}_2) =2$ and $\rho^2(\mathbf{A}_3) =3$, the condition is unsatisfied as $\kappa= 1.2>1$.	
\end{exmp}

\setcounter{thm}{1}
\begin{rem}
	Theorem~\ref{theory:stability} captures the stabilizability of the WNCS scheduling problem in terms of both the dynamic system parameters, $\mathbf{A}_i, \forall i\in \bar{\mathcal{F}}$, and the wireless channel conditions, i.e., $\{\bar{\xi}^s_{m,i},\bar{\xi}^c_{m,i}\}, i\in \bar{\mathcal{F}}, m=1,\dots,M$.
	Once \eqref{eq:stability} holds, there is at least one stationary and deterministic policy that stabilizes all plants of the WNCS.
	If the channel quality of different links does not differ much at the same frequency, then the sufficient stabilizability condition is tight.
	To the best of our knowledge, this is the first stabilizability condition established for $N$-plant-$M$-frequency WNCS with uplink and downlink scheduling in the literature.
\end{rem}

\section{Analysis and MDP Design} \label{sec:mainMDP}
As a performance measure  of the WNCS, 
we consider the expected (infinite horizon) total discounted cost (ETDC) given by
\begin{equation} \label{quadratic cost}
J = \sum_{k = 0}^{\infty} \vartheta^k
\sum_{i=1}^{N}
\mathbb{E}[\mathbf{x}^{\top}_{i,k}\mathbf{S}^x_{i}\mathbf{x}_{i,k} + \mathbf{u}^{\top}_{i,k}\mathbf{S}^u_{i}\mathbf{u}_{i,k}],
\end{equation}
where $\vartheta \in (0,1)$ is the discount factor, and a smaller $\vartheta$ means the future cost is less important.
$\mathbf{S}^x_{i}$ and $\mathbf{S}^u_{i}$ are positive definite  weighting matrices for the system state and control input of plant $i$, respectively.
Thus, it is important to find a scheduling policy that can minimize the design objective~\eqref{quadratic cost}.

Note that an ETDC minimization problem is commonly obtained by reformulating it into an MDP and solving it by classical policy and value iteration methods~\cite{bertsekas2000dynamic}. Theoretically speaking, the MDP solution provides an optimal deterministic and stationary policy, which is a mapping between its state and the scheduling action at each time step.
However, the optimal MDP
solution is intractable due to the uncertainties involved and the curse of dimensionality.
Thus, we seek to find a good approximate MDP solution by DRL. 
In the following, we aim to formulate the scheduler design problem into an MDP, and then present a DRL solution in Section~\ref{sec:RL}.

\par From \eqref{quadratic cost}, the per-step cost of the WNCS depends on state $\mathbf{x}_{i,k}$ and $\mathbf{u}_{i,k}$, which have continuous (uncountable) state spaces. Furthermore, $\mathbf{x}_{i,k}$ is not observable by the controller. To design a suitable MDP problem with a countable state space, we first need to determine an observable, discrete state of the MDP (in Section~\ref{sec:MDP}) and investigate how to represent the per-step cost function in \eqref{quadratic cost}, i.e., $\sum_{i=1}^{N}
\mathbb{E}[\mathbf{x}^{\top}_{i,k}\mathbf{S}^x_{i}\mathbf{x}_{i,k} + \mathbf{u}^{\top}_{i,k}\mathbf{S}^u_{i}\mathbf{u}_{i,k}]$, in terms of the state.

\subsection{MDP State Definition}
\begin{figure}[t]
	\centering\includegraphics[scale=0.8]{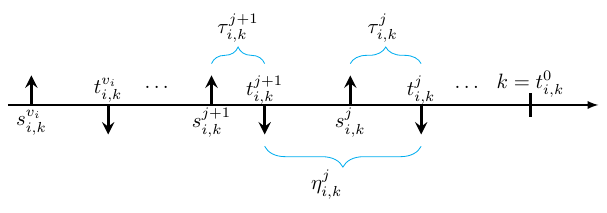}
	\vspace{-0.5cm}	
	\caption{Illustration of the state parameters of plant $i$.}
	\label{fig:AoI}
\end{figure}
We introduce event-related time parameters in the following.
Let $t_{i,k}^j, j = 1,2,\dots,v_i+1$, denote the time index of the $j$-th latest successful packet reception at actuator $i$ prior to the current time slot $k$, and $t_{i,k}^0 \triangleq k$. 
Let $s_{i,k}^j, j = 0,\dots,v_i$, denote the time index of the latest successful sensor $i$'s transmission prior to $t_{i,k}^j$, where $s_{i,k}^j\leq t_{i,k}^j$, as illustrated in Fig \ref{fig:AoI}.

Then, we define a sequence of variables, $\tau_{i,k}^j,j=0,\dots,v_i$, as
\begin{equation}\label{eq:all_tau}
\tau^j_{i,k} \triangleq t^j_{i,k}-s^j_{i,k},
\end{equation}
to record the estimation quality at $k$ and at the $v_i$ successful control transmissions, where $\tau_{i,k}^0 \triangleq \tau_{i,k}$ was defined above \eqref{tau_update}.

Similarly, we define  
\begin{equation}\label{eq:all_eta}
\eta_{i,k}^j \triangleq t_{i,k}^j - t_{i,k}^{j+1},\quad j=0,\dots,v_i,
\end{equation}
denoting the time duration between consecutive successful controller's transmissions, where $\eta_{i,k}^0 \triangleq \eta_{i,k}$ was defined above \eqref{eta_update}.
From \eqref{eq:all_tau} and \eqref{eq:all_eta}, $\tau^j_{i,k}$ and $\eta^j_{i,k}$ can be treated as the AoI of the sensor's and the controller's packet of plant $i$, respectively, observed at $t_{i,k}^j$.

From \eqref{eq:all_tau}, \eqref{tau_update}, \eqref{eq:all_eta}, and \eqref{eta_update}, the updating rules of $\tau_{i,k}^j$ and $\eta_{i,k}^j$ can be obtained as
\begin{equation} \label{tau_rules}
\tau_{i,k+1}^j = \begin{cases}
\left.
\begin{aligned}
&1, &\mbox{ if } \beta_{i,k} = 1\\
&\tau_{i,k}^0 + 1, &\mbox{ if } \beta_{i,k} = 0\\
\end{aligned}
\right\}\mbox{ for } j=0\\
\left.
\begin{aligned}
&\tau_{i,k}^{j-1}, &\mbox{ if } \gamma_{i,k+1} = 1\\
&\tau_{i,k}^j, &\mbox{ if } \gamma_{i,k+1} = 0\\
\end{aligned}
\right\}\mbox{ for } j=1,\cdots,v_i\\
\end{cases}
\end{equation}
\begin{equation} \label{eta_rules}
\eta_{i,k}^j = \begin{cases}
\left.
\begin{aligned}
&1, &\mbox{ if } \gamma_{i,k} = 1\\
&\eta_{i,k-1}^0 + 1, &\mbox{ if } \gamma_{i,k} = 0\\
\end{aligned}
\right\}\mbox{ for } j=0\\
\left.
\begin{aligned}
&\eta_{i,k-1}^{j-1}, &\mbox{ if } \gamma_{i,k} = 1\\
&\eta_{i,k-1}^j, &\mbox{ if } \gamma_{i,k} = 0\\
\end{aligned}
\right\}\mbox{ for } j=1,\cdots,v_i\\
\end{cases}
\end{equation}

Now we define the AoI-related vector state of plant $i$ as 
\begin{equation}
\label{eq:state}
\mathbf{s}_{i,k} \triangleq  (\tau_{i,k}^0,\cdots,\tau_{i,k}^{v_i},\eta_{i,k}^0,\cdots,\eta_{i,k}^{v_i}).
\end{equation}

\subsection{MDP Cost Function}
We will show that the per-step cost function of plant $i$ in \eqref{quadratic cost}, i.e., $\mathbb{E}[\mathbf{x}^{\top}_{i,k}\mathbf{S}^x_{i}\mathbf{x}_{i,k} + \mathbf{u}^{\top}_{i,k}\mathbf{S}^u_{i}\mathbf{u}_{i,k}]$, is determined by the vector state $\mathbf{s}_{i,k}$.
We focus on plant $i$ (the analytical method is identical for the other plants), and shall omit the plant index subscript of $\mathbf{x}_{i,k}$, $\mathbf{s}_{i,k}$, $\mathbf{u}_{i,k}$, $\mathbf{S}^x_{i}$ and $\mathbf{S}^u_{i}$ in the remainder of this subsection, where the corresponding parameters are $\mathbf{x}_{k}$, $\mathbf{s}_{k}$, $\mathbf{u}_{k}$, $\mathbf{S}^x$ and $\mathbf{S}^u$.

Taking \eqref{apply_u} into \eqref{plant_state}, the plant state evolution can be rewritten as
\begin{equation}\label{eq:x_iter}
\mathbf{x}_{k+1} = \mathbf{A}\mathbf{x}_k + \mathbf{B}\tilde{\mathbf{K}}(\mathbf{A}+\mathbf{B}\tilde{\mathbf{K}})^{\eta^0_k-1}\hat{\mathbf{x}}_{k+1-\eta^0_k} + \mathbf{w}_k.
\end{equation}
By using this backward iteration for $k-t^v_k$ times and the definition of $\eta_{k}^j$ and $t^j_k$, we have
\begin{equation}
\begin{aligned} \label{x_expression1}
&\mathbf{x}_{k} \!\!=\!\! (\mathbf{A}\!+\!\!\mathbf{B}\tilde{\mathbf{K}})^{\eta^0_k}\mathbf{x}_{t_k^1} \!\!+\!\! \big(\mathbf{A}^{\eta^0_k} \!\!-\!\! (\mathbf{A}\!+\!\mathbf{B}\tilde{\mathbf{K}})^{\eta^0_k}\big)\mathbf{e}_{t_k^1} \!\!+\!\!\! \sum_{i=1}^{\eta_k^0}\!\!\mathbf{A}^{i\!-\!1}\mathbf{w}_{k\!-i}\\
&= (\mathbf{A}+\mathbf{B}\tilde{\mathbf{K}})^{\eta^0_k}\times \\
& \bigg(\!\!(\mathbf{A}\!+\!\!\mathbf{B}\tilde{\mathbf{K}})^{\eta^1_k}\mathbf{x}_{t_k^2} \!\!+\! \big(\mathbf{A}^{\eta^1_k} \!-\! (\mathbf{A}\!+\!\mathbf{B}\tilde{\mathbf{K}})^{\eta^1_k}\big)\mathbf{e}_{t_k^2} \!+\!\! \sum_{i=1}^{\eta_k^1}\!\mathbf{A}^{i\!-\!1}\mathbf{w}_{t_k^1\!-\!i}\!\!\bigg)\\
&+ \big(\mathbf{A}^{\eta^0_k} - (\mathbf{A}+\mathbf{B}\tilde{\mathbf{K}})^{\eta^0_k}\big)\mathbf{e}_{t_k^1} + \sum_{i=1}^{\eta_k^0}\mathbf{A}^{i-1}\mathbf{w}_{k-i}\\
&=\dots\\
&= (\mathbf{A}+\mathbf{B}\tilde{\mathbf{K}})^{\eta_k^0 + \eta_k^1 + \cdots + \eta_k^{v-1}}\mathbf{x}_{t_k^v} + \mathbf{w}' + \mathbf{e}'\\
&= \mathbf{w}' + \mathbf{e}',
\end{aligned}
\end{equation}
where the last equation is due to the deadbeat control property~\eqref{control_property}. The quantity $\mathbf{w}'$ is related to the plant disturbance and $\mathbf{e}'$ is determined by the controller's estimation errors at the successful control packet transmissions, i.e., $\mathbf{e}_{t^1_k},\dots,\mathbf{e}_{t^{v}_k}$, as given below:
\begin{align}
&\mathbf{w}' = \sum_{i=1}^{\eta_k^0}\mathbf{A}^{i-1}\mathbf{w}_{k-i} + (\mathbf{A}+\mathbf{B}\tilde{\mathbf{K}})^{\eta_k^0}\sum_{i=1}^{\eta_k^1}\mathbf{A}^{i-1}\mathbf{w}_{k-i} +\\
& \cdots + (\mathbf{A}+\mathbf{B}\tilde{\mathbf{K}})^{\eta_k^0 + \cdots + \eta_k^{v-2}}\sum_{i=1}^{\eta_k^{v-1}}\mathbf{A}^{i-1}\mathbf{w}_{t_k^{v-1}-i} \nonumber \\
&= \sum_{j=1}^{v}\bigg((\mathbf{A}+\mathbf{B}\tilde{\mathbf{K}})^{\sum_{m=0}^{j-2}\eta_k^m}\sum_{i=t_k^j}^{t_k^{j-1}-1}\mathbf{A}^{t_k^{j-1}-1-i}\mathbf{w}_i\bigg) \label{w_expression},\\
&\mathbf{e}' = \big(\mathbf{A}^{\eta_k^0}-(\mathbf{A}+\mathbf{B}\tilde{\mathbf{K}})^{\eta_k^0}\big)\mathbf{e}_{t_k^1} +\\
& (\mathbf{A}+\mathbf{B}\tilde{\mathbf{K}})^{\eta_k^0}\big(\mathbf{A}^{\eta_k^1}-(\mathbf{A}+\mathbf{B}\tilde{\mathbf{K}})^{\eta_k^1}\big)\mathbf{e}_{t_k^2} \nonumber +\\
& \cdots + (\mathbf{A}\!+\!\mathbf{B}\tilde{\mathbf{K}})^{\eta_k^0 + \cdots + \eta_k^{v-2}}\big(\mathbf{A}^{\eta_k^{v-1}}\!-\!(\mathbf{A}+\mathbf{B}\tilde{\mathbf{K}})^{\eta_k^{v-1}}\big)\mathbf{e}_{t_k^v}. \label{error_collection}
\end{align} 
By analyzing the correlation between the sequences of plant disturbance and estimation noise, we have the following result.

\setcounter{thm}{0}
\begin{prop} \label{pro_3}
	The per-step cost function about the plant state is a deterministic function of the vector state $\mathbf{s}_k$ (see \eqref{eq:state}) as	
	\begin{equation}\label{eq:Jx}
	\normalfont
	J^x(\mathbf{s}_k) \triangleq \mathbb{E}\big[\mathbf{x}_{k}^{\top}\mathbf{S}^{x}\mathbf{x}_{k}\big] = \text{Tr}(\mathbf{S}^{x}V(\mathbf{s}_{k})),
	\end{equation}		
	where $V(\mathbf{s}_k) \triangleq \mathbb{E}[\mathbf{x}_k \mathbf{x}^\top_k]$ is the plant state covariance given in \eqref{P_k}. In the latter equation,  for $ 0 \leq i \leq j \leq v$, we have
	\begin{equation}\label{eq:delta}
	\Delta_k(i,j) \triangleq s_k^i - s_k^j = \sum_{n=i}^{j-1}\eta_k^{n} + \tau_k^j - \tau_k^i. 	
	\end{equation}
\begin{table*}	
	\begin{equation} \label{P_k}
	\begin{aligned}
	V(\mathbf{s}_k)& =  \mathbf{D}_k\hat{\mathbf{P}}^s(\mathbf{D}_k)^{\top} + \sum_{n=1}^{v-1}\big(\sum_{i=0}^{\Delta_k(n,n+1)-1}\check{\mathbf{E}}_{(k,n)}^i\mathbf{Q}_w(\check{\mathbf{E}}_{(k,n)}^i)^{\top}\big)
	+ \sum_{i=0}^{\eta_k^0+\tau_k^1-1}\mathbf{A}^i\mathbf{Q}_w(\mathbf{A}^i)^{\top} +  \sum_{n=1}^{v-1}\big(\sum_{i=0}^{\Delta_k(n,n+1)-1}\check{\mathbf{F}}^i_{(k,n)}\mathbf{Q}_v(\check{\mathbf{F}}^i_{(k,n)})^{\top}\big),\\
	\mathbf{D}_k &\triangleq \sum_{j=1}^{v}\bigg((\mathbf{A}+\mathbf{B}\tilde{\mathbf{K}})^{\sum_{m=0}^{j-2}\eta_k^m}\big(\mathbf{A}^{\eta_k^{j-1}}-(\mathbf{A}+\mathbf{B}\tilde{\mathbf{K}})^{\eta_k^{j-1}}\big)\mathbf{A}^{\tau_k^j}\mathbf{Z}^{\Delta_k(j,v)}\bigg), \\
\check{\mathbf{E}}_{(k,n)}^i &\triangleq (\mathbf{A}+\mathbf{B}\tilde{\mathbf{K}})^{\sum_{m=0}^{n-1}\eta_k^m}\mathbf{A}^{\tau_k^n + i}  
+ \sum_{j=1}^{n}\bigg((\mathbf{A}+\mathbf{B}\tilde{\mathbf{K}})^{\sum_{m=0}^{j-2}\eta_k^m}\big(\mathbf{A}^{\eta_k^{j-1}}-(\mathbf{A}+\mathbf{B}\tilde{\mathbf{K}})^{\eta_k^{j-1}}\big)\mathbf{A}^{\tau_k^j}\mathbf{Z}^{i + \Delta_k(j,n)}(\mathbf{I}-\hat{\mathbf{K}}\mathbf{C})\bigg), \\
\check{\mathbf{F}}^i_{(k,n)} &\triangleq \sum_{j=1}^{n}\bigg((\mathbf{A}+\mathbf{B}\tilde{\mathbf{K}})^{\sum_{m=0}^{j-2}\eta_k^m}\big(\mathbf{A}^{\eta_k^{j-1}}-(\mathbf{A}+\mathbf{B}\tilde{\mathbf{K}})^{\eta_k^{j-1}}\big)\mathbf{A}^{\tau_k^j}\mathbf{Z}^{\Delta_k(j,n)+i}\hat{\mathbf{K}}\bigg)
	\end{aligned}
	\end{equation}

\end{table*}
\end{prop}
\vspace{-0.7cm}
Then, building on the system dynamics~\eqref{plant_state}, the local estimate~\eqref{eq:KF}, the remote estimate~\eqref{eq:RE}, and the control input~\eqref{apply_u}, and by comprehensively analyzing the effect of the correlations between plant disturbance and estimation noise on the control input covariance,  we obtain the per-step cost function about the control input as below.
\begin{prop} \label{pro_4}
	The per-step cost function about the control input at $k$ is a deterministic function of the vector state $\mathbf{s}_{k}$ as
	\begin{equation} \label{eq:Ju}
	\normalfont
	\begin{aligned}
	&J^u(\mathbf{s}_{k}) \triangleq \mathbb{E}\big[\mathbf{u}_k^{\top}\mathbf{S}^u\mathbf{u}_k\big]\\
	&=\!\! \text{Tr}\bigg(\!\!\big(\tilde{\mathbf{K}}(\!\mathbf{A}\!+\!\mathbf{B}\tilde{\mathbf{K}}\!)^{\eta_k\!-\!1}\big)^{\!\!\top}\!\mathbf{S}^u\big(\tilde{\mathbf{K}}(\!\mathbf{A}\!+\!\mathbf{B}\tilde{\mathbf{K}}\!)^{\eta_k\!-\!1}\big)\hat{V}(
	\mathbf{s}_{k+1-\eta_k}
	)\!\!\bigg),
	\end{aligned}
	\end{equation}
	where $\mathbf{s}_{k+1-\eta_k}$ can be directly obtained by $\mathbf{s}_{k}$. $\hat{V}(\mathbf{s}_k) \triangleq  \mathbb{E}[\hat{\mathbf{x}}_k \hat{\mathbf{x}}_k^\top]$ is the covariance of the remote estimate given in~\eqref{eq:V_hat}, where
	\begin{align}\label{eq:prop2_DEF}
\tilde{\mathbf{D}}_k\ &\triangleq \mathbf{D}_k - \mathbf{A}^{\tau_k^0}\mathbf{Z}^{\Delta_k(0,v)},\\
\dot{\mathbf{E}}_{(k,n)}^i &\triangleq \check{\mathbf{E}}_{(k,n)}^i - \mathbf{A}^{\tau_k^0}\mathbf{Z}^{\Delta_k(0,n)-\tau_k^n+i}(\mathbf{I}-\hat{\mathbf{K}}\mathbf{C}),\\
\dot{\mathbf{F}}_{(k,n)}^i &\triangleq \check{\mathbf{F}}_{(k,n)}^i + \mathbf{A}^{\tau_k^0}\mathbf{Z}^{\Delta_k(0,n)+i}\hat{\mathbf{K}}.
\end{align}

\begin{table*}
		\begin{equation}\label{eq:V_hat}
	\begin{aligned}
	&\hat{V}(\mathbf{s}_k) = \tilde{\mathbf{D}}_k\hat{\mathbf{P}}^s\tilde{\mathbf{D}}_k^{\top} \!\!+\!\! \sum_{n=1}^{v-1}\big(\sum_{i=0}^{\Delta_k(n,n+1)-1}\dot{\mathbf{E}}_{(k,n)}^i\mathbf{Q}_w(\dot{\mathbf{E}}_{(k,n)}^i)^{\top}\big)
	+ \!\!\!\!\!\!\!\!\sum_{i=0}^{\Delta_k(0,1)-1}\!\!\big(\mathbf{A}^{\tau_k^0 + i} - \mathbf{A}^{\tau_k^0}\mathbf{Z}^{i}(\mathbf{I}-\hat{\mathbf{K}}\mathbf{C})\big)\mathbf{Q}_w\big(\mathbf{A}^{\tau_k^0 + i} - \mathbf{A}^{\tau_k^0}\mathbf{Z}^{i}(\mathbf{I}-\hat{\mathbf{K}}\mathbf{C})\big)^{\top}\\
	&\quad \qquad + \sum_{n=1}^{v-1}\big(\sum_{i=0}^{\Delta_k(n,n+1)-1}\dot{\mathbf{F}}_{(k,n)}^i\mathbf{Q}_w(\dot{\mathbf{F}}_{(k,n)}^i)^{\top}\big) + \sum_{i=0}^{\Delta_k(0,1)-1}\big(\mathbf{A}^{\tau_k^0}\mathbf{Z}^{i}\hat{\mathbf{K}}\big)\mathbf{Q}_w\big(\mathbf{A}^{\tau_k^0}\mathbf{Z}^{i}\hat{\mathbf{K}}\big)^{\top}
	\end{aligned}
	\end{equation}

\end{table*}
\end{prop}
The proofs of Propositions~\ref{pro_3} and~\ref{pro_4} are given in~\cite{liu2021DRL} due to the space limitation.
\setcounter{thm}{2}
\begin{rem}
	From Propositions~\ref{pro_3} and~\ref{pro_4}, the per-step cost of the plant is determined by the finite-length vector state $\mathbf{s}_k$. However, expressions for the cost functions are involved due to the sequential predictive control and the command buffer adopted at the actuator, as well as the consideration of plant disturbance, and local and remote estimation errors.
\end{rem}
By substituting plant index $i$ into the vector state and the cost functions \eqref{eq:Jx} and \eqref{eq:Ju},  the above results have opened the door to address the optimal scheduling problem of the $N$-plant-$M$-frequency WNCS with ETDC in~\eqref{quadratic cost} as a decision making problem with a countable state space.

\subsection{Resulting MDP}\label{sec:MDP}
From \eqref{eq:prob}, \eqref{tau_rules} and \eqref{eta_rules}, given the current state $\mathbf{s}_{i,k}$ and the current transmission scheduling action related to plant $i$, the next state, $\mathbf{s}_{i,k+1}$, is independent of all previous states and actions, satisfying the Markov property.
Thus,  the transmission scheduling problem of the WNCS can be formulated as an MDP:
\begin{itemize}
	\item The state of the MDP at time $k$ is
	$\mathbf{s}_k \triangleq (\mathbf{s}_{1,k}, \cdots, \mathbf{s}_{N,k})$, where $\mathbf{s}_{i,k}$ is as defined in \eqref{eq:state}, $i =1,\dots,N$.
	The state space is $\mathcal{S} = \underbrace{\mathbb{N}^{2v_1+2} \times \cdots \times \mathbb{N}^{2v_N+2}}_{N \; \text{terms}}$. 
	
	\item The action at time $k$, $\mathbf{a}_k  \triangleq  [a_{1,k},a_{2,k},\dots,a_{M,k}] =\pi(\mathbf{s}_k)$, is the transmission link allocation at each frequency, where $a_{m,k}\in\{-N,\dots,N\},\ m=1,\dots,M$, and	
	${a}_{m,k}\neq {a}_{m',k}$ if ${a}_{m,k}{a}_{m',k}\neq0$.
	Then, the action space $\mathcal{A} \subseteq \{-N,\dots,N\}^M$ has the cardinality of $|\mathcal{A}| = \sum_{m=0}^{M} \mathsf{C}^m_M \mathsf{P}^m_{2N}$. 
	
	\item The state transition probability $P(\mathbf{s}_{k+1}|\mathbf{s}_k,\mathbf{a}_k)$ can be obtained directly from the state updating rules in \eqref{eq:prob}, \eqref{tau_rules} and \eqref{eta_rules}.
	
	\item The per-step cost of the MDP is the sum cost of each plant in Propositions~\ref{pro_3} and~\ref{pro_4} as
	\begin{equation}\label{eq:per-step-cost}
	c(\mathbf{s}_k) \triangleq \sum_{i=1}^{N} c_i(\mathbf{s}_{i,k}),
	\end{equation}
	where $c_i(\mathbf{s}_{i,k}) \triangleq J^x_{i}(\mathbf{s}_{i,k}) + J^u_{i}(\mathbf{s}_{i,k})$.
	
	\item The discount factor is $\vartheta \in (0,1)$.
	
	\item The scheduling problem of the $N$-plant-$M$-frequency system can be rewritten as
	\begin{equation}\label{eq:problem}
	\min_{\pi} \mathbb{E}^\pi\left[\sum_{k=0}^{\infty} \vartheta^k c(\mathbf{s}_k)\right].
	\end{equation}
\end{itemize}

\begin{rem} \label{re:complexity}
	The discounted MDP problem above with an infinite state space can be numerically solved to some extent by using  classic policy or value iteration methods with a truncated state space $\mathcal{S}_L$ 
	\begin{equation}
	\mathcal{S}_L \triangleq \underbrace{\mathbb{N}_L^{2v_1+2} \times \cdots \times \mathbb{N}_L^{2v_N+2}}_{N \; \textrm{terms}}, \text{where } \mathbb{N}_L = \{1,\dots, L\}.
	\end{equation}
	The computation complexity of relative value iteration algorithm is given as $\mathcal{O}(|\mathcal{S}_L||\mathcal{A}|^2K)$~\cite{sennott2009stochastic}, where $K$ is the number of iteration steps, and the state space and action space sizes are $|\mathcal{S}| = L^{2\sum_{i=1}^{N}v_i+2N}$ and $|\mathcal{A}| = \sum_{m=0}^{M} \mathsf{C}^m_M \mathsf{P}^m_{2N}$, respectively.
	However, the sizes of both the state and action spaces are considerably large even for relatively small $N$ and $M$, leading to  numerical difficulties in finding a solution.
	In the literature of WNCS, even for the (simpler) $N$-plant-$M$-channel remote estimation system, only the $M=1$ case has been found to have a numerical  solution~\cite{Wu2018Auto}.
	To tackle the challenge for larger scale WNCS deployment, we will use DRL methods
	exploiting DNNs for function approximation in the following.
\end{rem}

\begin{rem}
Although the MDP problem formulation of the WNCS assumes static wireless channels with fixed packet drop probabilities, it can be extended to a Markov fading channel scenario. A Markov fading channel can have multiple channel states with different packet drop probabilities, and the channel state transition is modeled by a Markov chain~\cite{liu2020remote}. 
In this scenario, the state of the MDP problem should also include the channel state, and the state transition probability needs to take into account both the AoI state and the Markov channel state transition probabilities.
The details of WNCS scheduling over Markov fading channels can be investigated in our future work.
\end{rem}

\section{WNCS Scheduling with Deep Reinforcement Learning}\label{sec:RL}
Building on the MDP framework, DRL is widely applied in solving decision making problems with pre-designed state space, action space, per-step reward (cost) function and discount factor for achieving the maximum long-term reward \cite{10931147,10917000}. 
The main difference is that DRL does not exploit the state transition probability as required by MDP, but records and utilizes many sampled data sequences, including the current state $\mathbf{s}$, action $\mathbf{a}$, reward $r$ and next state $\mathbf{s}'$, to train DNNs for generating the optimal policy \cite{bertse19a,10876598}.
To find deterministic policies\footnote{The present work focuses on deterministic scheduling policies as it has been proved that the optimal policy is deterministic for unconstrained MDP problems~\cite{bertsekas2000dynamic}.}, the most widely considered DRL algorithms are DQN~\cite{DQN}, DDPG~\cite{lillicrap2015continuous}, and TD3~\cite{fujimoto2018addressing}. In this study, we applied three algorithms to address the scheduling problem.

DQN, which is the focus of our detailed algorithmic presentation, is the foundational algorithm in the field of DRL, primarily designed for environments with discrete action spaces. DDPG extends the ideas from DQN by adapting them to continuous action spaces using policy gradient methods. TD3, further building on DDPG, introduces improvements in terms of stability and performance. DQN serves as a foundation for understanding the basic principles that are also applicable to the more complex algorithms like DDPG and TD3, and is well-suited for our problem (Section~\ref{sec:MDP}) with the discrete and finite action space $\mathcal{A}$. In this regard, we choose to describe the DQN algorithm in detailed technical depth to avoid redundancy. Detailed pseudocode for DDPG and TD3, while omitted from the main text to conserve space, is available in~\cite{lillicrap2015continuous,fujimoto2018addressing} for interested readers. 

Compared with standard DQN (see Section~\ref{sec:origional}), the primary novelty of our approach lies in how it tackles the rapid growth of the scheduling action space. First, we restrict each plant to either uplink or downlink mode, guided by the practical observation that sensor information must be received before control commands are updated (see Section~\ref{sec:reduced}). Second, for scenarios involving many plants and frequencies, we embed these reduced actions into a continuous space by assigning a real-valued ‘priority score’ to each link (see Section~\ref{sec:ActionEmbed}). This allows advanced DRL algorithms like DDPG or TD3 to be applied. Consequently, although we build upon the established DQN methodology in Section~\ref{sec:origional}, our innovations in action space reduction and action embedding enable us to solve large-scale scheduling problems more effectively than a naive DQN approach could.

In general, we consider an episodic training scenario, where the system is operating in a simulated environment. 
Hence, one can learn and gather samples over multiple episodes to train a policy until convergence. Once training is completed offline, the properties of the final scheduling policy are verified so that one can “safely” deploy it at the central controller—rather than the sensors for real-time operation~\cite {lillicrap2015continuous}. Online operation requires only a simple forward pass through the pre-trained neural network—a computationally lightweight procedure. The sensors merely provide state measurements and other necessary feedback, thereby avoiding any undue computational burden for sensors.
We note that if the DRL policy does lead to instability (i.e., an unbounded expected cost), one needs to adjust the hyperparameters of the DNN (e.g., the initial network parameters and the number of layers and neurons) and start retraining. Such readjusting of hyperparameters of a deep learning agent is commonly adopted in practice~\cite{DQN}.

In the following, we present a deep-Q-learning approach for scheduler design that requires a sampled data sequence $\{(\mathbf{s,a},r,\mathbf{s'})\}$.
We note that the data sampling and the deep Q-learning are conducted offline. 
In particular, given the current state and action pair $(\mathbf{s,a})$, the reward $r$ can be obtained immediately from Propositions~\ref{pro_3} and~\ref{pro_4}, and the next state $\mathbf{s}'$ affected by the scheduling action and the packet dropouts can be sampled easily based on the packet error probabilities of each channel~\eqref{eq:prob} and state transition rules~\eqref{tau_rules} and \eqref{eta_rules}.
Furthermore, the trained policy needs to be tested offline to verify the WNCS's stability before an online deployment.

\subsection{Deep Q-Learning Approach}\label{sec:origional}
We introduce the state-action value function given a policy $\pi(\cdot)$, which is also called the $Q$ function~\cite{bertsekas2000dynamic}:
\begin{equation}
Q^\pi(\mathbf{s,a})\! \triangleq\! \mathbb{E}\left[\sum_{k=0}^{\infty} \vartheta^k r_k \vert \mathbf{s}_0=\!\mathbf{s}, \mathbf{a}_0=\!\mathbf{a},\mathbf{a}_{k}=\!\pi(\mathbf{s}_k),\forall k>0\right]\!,
\end{equation}
where $r_k \triangleq -c(\mathbf{s}_k)$ can be treated as the negative cost function in~\eqref{eq:per-step-cost} at $k$.
Then, by dropping out the time index $k$, the $Q$ function of the maximum average discounted total reward achieving policy $\pi^\star(\cdot)$ satisfies the Bellman equation as
\begin{equation}\label{eq:bellman}
Q^{\star}(\mathbf{s,a}) = \mathbb{E}\big[r + \vartheta \max_{\mathbf{a}'\in \mathcal{A}}Q^{\star}(\mathbf{s', a'})|\mathbf{s,a}\big].
\end{equation}
The optimal stationary and deterministic policy can be written as
\begin{equation}
\mathbf{a}^{\star}=\pi^\star(\mathbf{s}) \triangleq \arg\max_{\mathbf{a}\in \mathcal{A}}Q^{\star}(\mathbf{s,a}).
\end{equation}

Solving the optimal $Q$ function is the key to finding the optimal policy, but it is computationally intractable by conventional methods as discussed in Remark~\ref{re:complexity}. 
In contrast, deep Q-learning methods approximate $Q^\star(\mathbf{s, a})$ by 
a function $Q(\mathbf{s, a};\theta)$ parameterized by a set of neural network parameters,  $\theta$, (including both weights and biases), and then learns $\theta$ to minimize the difference between the left- and right-hand sides of~\eqref{eq:bellman}~\cite{sutton2018reinforcement}.
Deep Q-learning can be easily implemented by the most well-known machine learning framework, TensorFlow, with experience replay buffer, $\epsilon$-greedy exploration and mini-batch sampling techniques~\cite{TensorFlow}.

\subsection{Deep Q-Learning with Reduced Action Space}\label{sec:reduced}
In practice, the training of DQN converges and the trained DQN performs well for simple decision-making problems with a bounded reward function, small input state dimension and small action space, see e.g.,~\cite{TensorFlow}, where the state input is a length-$4$ vector and there are only $2$ possible actions. However, for the problem of interest in the current work, the reward function is unbounded (due to the potential of consecutive packet dropouts), the state  $\mathbf{s}$ is of high dimension, and action-space $|\mathcal{A}|$ is large.  Hence, the convergence of the DQN training is not guaranteed for all hyper-parameters, which include the initialization of $\theta$, the number of neural network layers, the number of neurons per layer, and the replay buffer and mini-batch sizes.\footnote{Note that when $N=3$, $M=3$ and $v_1=\dots=v_N=2$, the length of the   state is $18$ and the action space size is $229$ based on Section~\ref{sec:MDP}.}
Even if the convergence is achieved, due to the complexity of the problem, it may take very long training time and only converge to a local optimal parameter set $\theta$ (which is also affected by the choice of hyper-parameters), leading to a worse performance than some conventional scheduling policies. Whilst the hyper-parameters can, in principle, be chosen to enhance performance, no appropriate tuning guidelines exist for the problem at hand. 

\par To overcome the computational issues outlined above, in the following we will reduce the action space. This will make the training task simpler and enhance the convergence rate. To be more specific, for each plant system, instead of considering all the possible actions to schedule either or both of the uplink and downlink communications at each time instant, we restrict to the two modes \emph{downlink only} and \emph{uplink only}. Switching between these two modes of operation occurs once a scheduled transmission of that plant is successful.
\par The principle behind this is the intuition that the controller requires the information from the sensor first and then performs control, and so on and so forth.
The above concept can be incorporated into the current framework by adding the downlink/uplink indicator $s^{\text{link}}_{i,k}\in\{-1,1\},\forall i =1,\dots,N$,  into state $\mathbf{s}_k$ leading to the aggregated state $$\breve{\mathbf{s}}_k \triangleq [\mathbf{s}_k,s^{\text{link}}_{1,k},\dots,s^{\text{link}}_{N,k}],$$ where $s^{\text{link}}_{i,k} = 1$ or $-1$ indicates the uplink or downlink of plant $i$ can be scheduled at $k$. The state updating rule is
\begin{equation}\label{eq:s_link}
s^{\text{link}}_{i,k+1} = \begin{cases}
-1 &\text{ if } s^{\text{link}}_{i,k}=1 \text{ and } \beta_{i,k} =1\\
1 &\text{ if } s^{\text{link}}_{i,k}=-1 \text{ and } \gamma_{i,k} =1\\
s^{\text{link}}_{i,k} &\text{ otherwise.} \\
\end{cases}
\end{equation}

The new scheduling action at the $M$ frequencies is denoted as $\breve{\mathbf{a}}_k \triangleq [\breve{a}_{1,k},\dots,\breve{a}_{M,k}]$, where $\breve{a}_{m,k} \in \{0,\dots,N\},\; \forall m\in\{1,\dots,M\}$, and $\breve{a}_{m,k}\neq \breve{a}_{m',k}$ if $\breve{a}_{m,k}\breve{a}_{m',k}\neq0$.
If $\breve{a}_{m,k} =i$ and $s^{\text{link}}_{i,k} =1$, then the uplink of plant $i$ is scheduled on frequency $m$ at time $k$.
Compared with the original action $\mathbf{a}_k$ in Section~\ref{sec:MDP}, the size of the new action space is reduced significantly to $|\tilde{\mathcal{A}}| = \sum_{m=0}^{M} \mathsf{C}^m_M \mathsf{P}^m_{N}$.
For example, 
when $N=3$, $M=3$ and $v_1=v_2=v_3=2$, the action space size is reduced from $229$ to $34$.
Building on this reduced action space, the approach for solving problem~\eqref{eq:problem} is given as Algorithm~\ref{algor:1}.
We note that the reduced action space still rapidly increases with the increment of $N$ and $M$. This makes it difficult for the DQN-based algorithm to solve the scheduling problem of larger-scale WNCS, such as $N = 10$ and $M = 10$ with a discrete action space of $234662231$.

\begin{algorithm}[t] 
	\begin{algorithmic}[1] 
		\State Initialize experience replay buffer $\mathcal{B}$ to capacity $K$
		\State Initialize multi-layer fully connected  neural network $Q$ with vector input $\breve{\mathbf{s}}$, $|\tilde{\mathcal{A}}|$ outputs $\{Q(\breve{\mathbf{s}},\breve{\mathbf{a}}_1;\theta_0),\dots,Q(\breve{\mathbf{s}},\breve{\mathbf{a}}_{|\mathcal{A}|};\theta_0)\}$ and random parameter set $\theta_0$	
		\For{episode = $1, \cdots, E$} 
		\State Randomly initialize $\mathbf{s}_0$, and convert $\mathbf{s}_0$ to $\breve{\mathbf{s}}_0$ 
		\For{$t = 0, 1, \cdots, T$}
		\State With probability $\epsilon$ select a random action $\breve{\mathbf{a}}_t$, otherwise select $\breve{\mathbf{a}}_t = \arg\max_{\breve{\mathbf{a}}\in\tilde{\mathcal{A}}}Q(\breve{\mathbf{s}}_t, \breve{\mathbf{a}};\theta_t)$
		\State Convert $\breve{\mathbf{a}}_t$ to $\mathbf{a}_t$ 
        \State Execute $\mathbf{a}_t$, and obtain $r_t$ and $\mathbf{s}_{t+1}$
        \State Convert $\mathbf{s}_{t+1}$ to $\breve{\mathbf{s}}_{t+1}$
		\State Store $(\breve{\mathbf{s}}_t, \breve{\mathbf{a}}_t, r_t, \breve{\mathbf{s}}_{t+1})$ in $\mathcal{B}$
		\State Sample random mini-batch of $l$ transitions $(\breve{\mathbf{s}}_t, \breve{\mathbf{a}}_t, r_t, \breve{\mathbf{s}}_{t+1})$ from $\mathcal{B}$ as $\tilde{\mathcal{B}}$
		\State Set $z_j = r_j + \vartheta \max_{\breve{\mathbf{a}}'\in \tilde{\mathcal{A}}}\hat{Q}(\breve{\mathbf{s}}_{j+1}, \breve{\mathbf{a}}';\theta_t)$ for each sample in $\tilde{\mathcal{B}}$
		\State Perform a mini-batch gradient descend step to minimize the Bellman error, i.e., $\min_{\hat{\theta}}\sum_{\tilde{\mathcal{B}}}(z_j - Q(\breve{\mathbf{s}}_j, \breve{\mathbf{a}}_j;\hat{\theta}))^2$
		\State Update $\theta_{t+1}=\hat{\theta}$
		\EndFor 
		\EndFor 
	\end{algorithmic} 
	\caption{Deep Q-learning with reduced action space for transmission scheduling in WNCS} 
	\label{algor:1}
\end{algorithm}

\subsection{Action Embedding for DDPG and TD3} \label{sec:ActionEmbed}
To further handle the issue of large discrete action space, the discrete actions can be embedded into a continuous action space. This approach uses a vector of continuous values to represent the priority of allocating channels to plants. Specifically, each plant is assigned a continuous value. Plants are then allocated channels based on the ranking of these continuous values. For instance, the plant with the highest value is allocated to the first channel, the second highest to the second channel, and so on. This method significantly reduces the complexity of the discrete action space, making it feasible to handle larger $N$ and $M$. For example, the number of output neurons is $6$ when $N = 6$ and $M = 4$, while the original discrete action space is $1045$. The continuous action space allows for more efficient learning and optimization by leveraging the capabilities of DRL algorithms designed for continuous spaces.
We will illustrate the numerical results of using DRL for solving the scheduling problem in the following section.

\section{Numerical Results}
\subsection{Experiment Setup}
We adopt the following hyper-parameters for the DRL algorithms. The input state dimension of DQN is set to $2N(v_i+1)$. We use three hidden layers with $300$, $200$, $100$ neurons, respectively. The output layer of DQN has $|\tilde{\mathcal{A}}| = \sum_{m=0}^{M} \mathsf{C}^m_M \mathsf{P}^m_{N}$ outputs (actions) for the problem with reduced action space (see Section~\ref{sec:reduced}).
The activation function in each hidden layer is the ReLU~\cite{TensorFlow}. The activation function in the output layer of DDPG and TD3 is the Sigmoid~\cite{TensorFlow}. 
 The experience replay memory has a size of $K=100000$, and the mini-batch size is $64$. 
So in each time step, $64$ data is sampled from the replay memory for training. The exploration parameter $\epsilon$ of DQN decreases from $1$ to $0.01$ at the rate of $0.999$ after each time step. 
The exploration method for DDPG and TD3 is based on Gaussian action noise with a standard deviation of $0.2$ and a mean value of $0$.
We adopt the ADAM optimizer to update the DNNs.
Each training takes $E= 500$ episodes, each including $T = 500$ time steps. We test the trained policy for $100$ episodes and compare the empirical average costs $\frac{1}{T} \sum_{k=1}^{T} \sum_{i=1}^{N}c(\mathbf{s}_{i,k})$.

We use Algorithm~\ref{algor:1} (DQN) and its adapted version, i.e., DDPG and TD3, to solve the scheduling problem with reduced action space and compare it with the original DQN, DDPG and TD3, and three heuristic benchmark policies.
1) \textbf{Random policy}: randomly choose $M$ out of the $2N$ links and randomly allocate them to the $M$ frequencies at each time.
2) \textbf{Round-robin policy}: the $2N$ links are divided into $M$ groups, and the links in each of the groups $m\in\{1,\dots,M\}$ are allocated to the corresponding frequency $m$; the transmission scheduling at each frequency follows the round-robin fashion~\cite{LEONG2020108759} at each time. Note that in the simulation, we test all link-grouping combinations and only present the one with the lowest average cost.
3) \textbf{Greedy policy}: At each time, the $M$ frequencies are randomly allocated to $M$ of the $2N$ links with the largest AoI value in $\{\tau_{1,k},\eta_{1,k},\dots,\tau_{N,k},\eta_{N,k}\}$. Recall that $\tau_{i,k}$ and $\eta_{i,k}$ denote the time duration since the last received sensing and control packets of plant $i$, respectively.

\subsection{Performance Comparison}
\begin{figure}[t]
	\centering\includegraphics[scale=0.48]{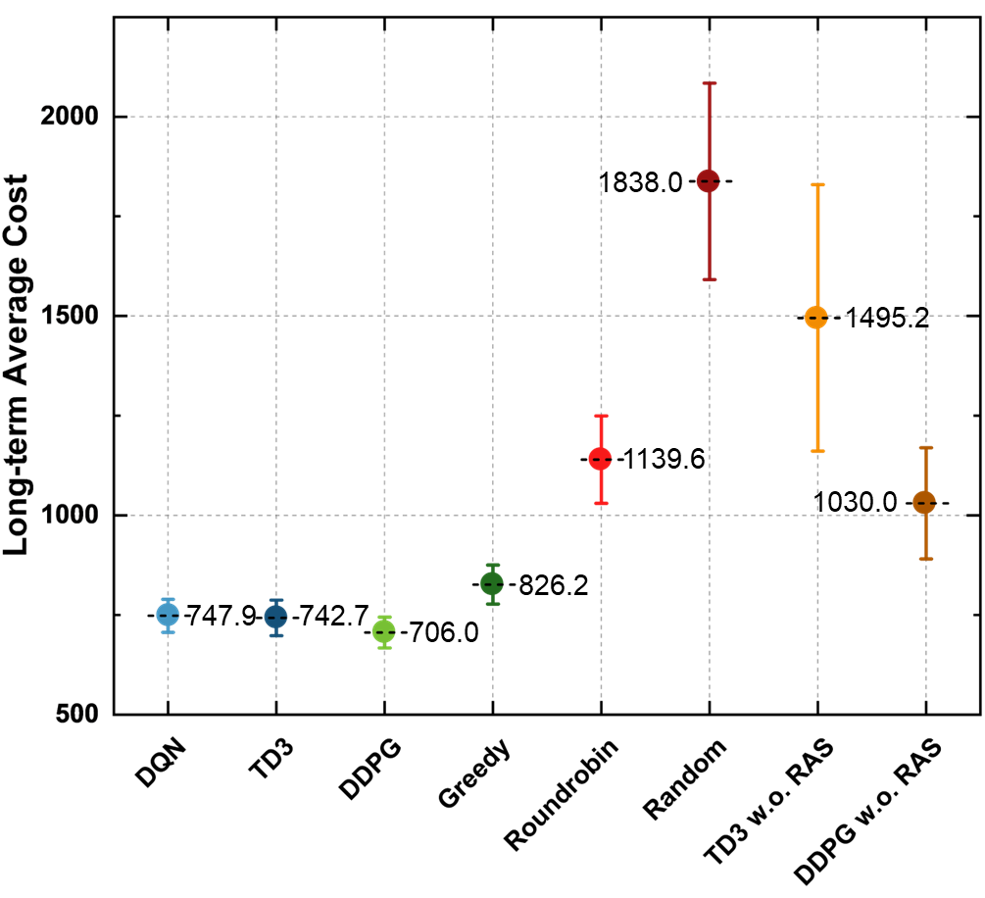}
	\caption{The long-term average performance of DRL and benchmark algorithms over a system with $N = 5$ and $M = 5$. Values are means $\pm$ standard error of mean. RAS denotes reduced action space.}
	\label{fig:sim1}
\end{figure}
In the comparison experiment, the plant system matrices $\mathbf{A}_i$ are randomly generated, in which case $\rho(\mathbf{A}_i) \in (1.0,1.1)$.
The control input matrices are all set to $\mathbf{B}_i = [1\ 1]^\top$. The measurement matrix $\mathbf{C}_i$ is equal to the identity matrix. The covariance matrices are $\mathbf{Q}^w_i = \mathbf{Q}^v_i = 0.1 \mathbf{I}$. Each plant is $2$-step controllable, i.e., $v_i=2$.
The packet success probabilities of each link, i.e., $\xi^s_{m,i}$ and $\xi^c_{m,i}$, are generated randomly and drawn uniformly from $(0.5,1.0)$. 
The weighting terms $\mathbf{S}^x_{i}$ and $\mathbf{S}^u_{i}$ are chosen as identity matrices and the discount factor is $\vartheta = 0.95$.
The resulting WNCS satisfies the stability condition established in Theorem \ref{theory:stability}.
We first consider a $N$-plant-$M$-frequency WNCS with $N = 5$ and $M = 5$, i.e., $10$ links sharing $5$ frequencies with a reduced discrete action space of $1546$. Then, we will extend to larger scales, such as $N = 10$ and $M = 10$, i.e., $20$ links sharing $10$ frequencies with a reduced discrete action space of $234662231$.

Fig.~\ref{fig:sim1} showcases the comparison of three DRL algorithms and the three heuristic benchmark policies when $N = 5$ and $M = 5$. All the DRL algorithms with the reduced action space, including DQN, DDPG, and TD3, achieve a lower long-term average cost than the best heuristic benchmark policy, i.e., Greedy Policy. We note that DRL is a dynamic decision-making algorithm designed to maximize the cumulative reward over time, inherently aligning with the goal of optimizing long-term costs. It evaluates actions based on their expected future rewards, considering both immediate and delayed consequences. Heuristic benchmark policies are based on static decision rules and lack the ability to adapt to changing environments or evolving system dynamics. They often prioritize immediate rewards or costs, making decisions that optimize short-term performance without considering long-term implications. They might select actions that yield immediate gains but could lead to suboptimal long-term outcomes, such as increased future costs. Thus, DRL algorithms show positive performance.

\begin{figure}[t]
	\centering\includegraphics[scale=0.48]{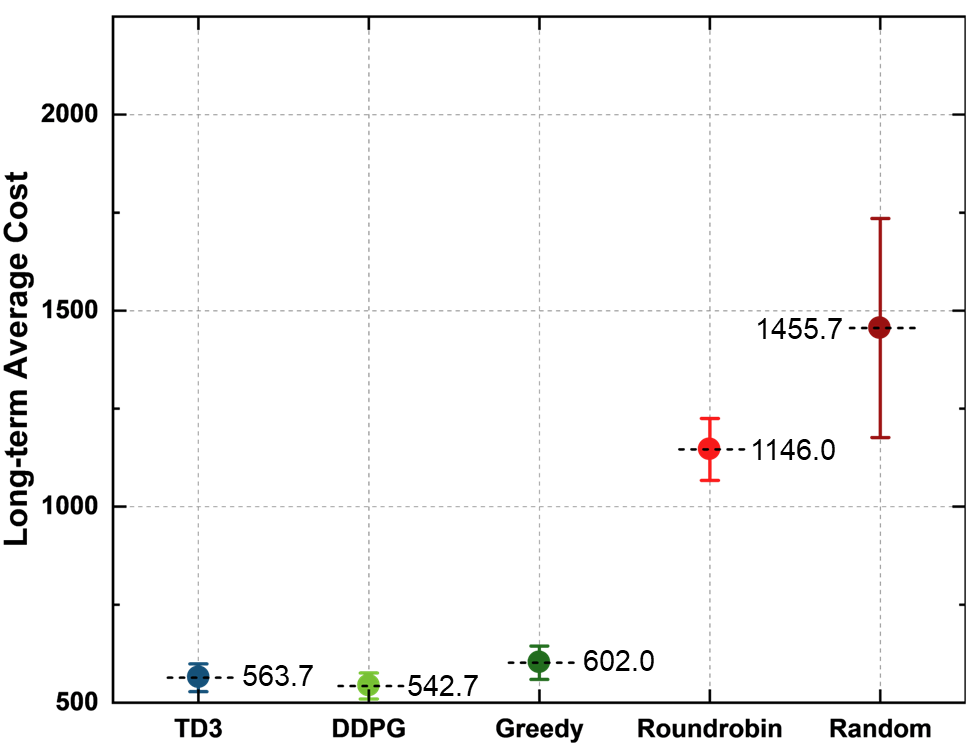}	
	\caption{The long-term average performance of DRL and benchmark algorithms over a system with $N = 10$ and $M = 10$. Values are means $\pm$ standard error of mean.}
	\label{fig:sim2}
\end{figure}

In addition, Fig.~\ref{fig:sim1} demonstrates that the well-trained DQN-based policy reduces the long-term average cost of greedy policy by $9.5$\% from $826.2$ to $747.9$. The DDPG and TD3 algorithms based on the action embedding method achieve similar or slightly better performance compared to the DQN-based algorithm. This is attributed to the large discrete action space, which makes it difficult for the DQN-based algorithm to explore adequately during training and to get the optimal policy after training. The DDPG algorithm reduces the long-term average cost of the greedy policy by $14.5$\% from $826.2$ to $706.0$. TD3 is an extension of DDPG, which is designed to reduce overestimation bias and variance in the learning process by introducing delayed policy updates, target policy smoothing, and twin critic networks. These features can lead to overly conservative behavior in WNCS with high stochasticity, which potentially slows down the learning process or leads to underfitting where the system fails to adequately capture the beneficial actions amidst the noise. Thus, DDPG, as a more aggressive approach, is better than TD3.

Fig.~\ref{fig:sim1} also shows that the DQN algorithm without reduced action space (RAS) has an action space size of $63591$, which is unable to solve the original MDP problem with such a large action space. The DDPG and TD3 algorithms without RAS cannot outperform the ones with RAS. These results illustrate the effectiveness of the proposed action space reduction method.

Fig.~\ref{fig:sim2} and Fig.~\ref{fig:sim3} showcase the comparison of three DRL algorithms and the three heuristic benchmark policies when the system scale is large. In these cases, the DQN-based algorithm is unable to solve the MDP, as the reduced discrete action space is still too large. Results show that the DDPG and TD3 algorithms based on the action embedding method effectively address the challenge of the rapidly growing action space, making it more suitable for large-scale WNCS problems. The DDPG algorithm reduces the long-term average cost of the greedy policy by $9.9$\% from $602.0$ to $542.7$ when $N = 10$ and $M = 10$. As the communication resources become relatively scarce, the performance gain dramatically increases to $39.9$\% (from $515.9$ to $309.6$), when $N = 8$ and $M = 6$, demonstrating the effectiveness of the proposed framework for more challenging scenarios.

\begin{figure}[t]
	\centering\includegraphics[scale=0.48]{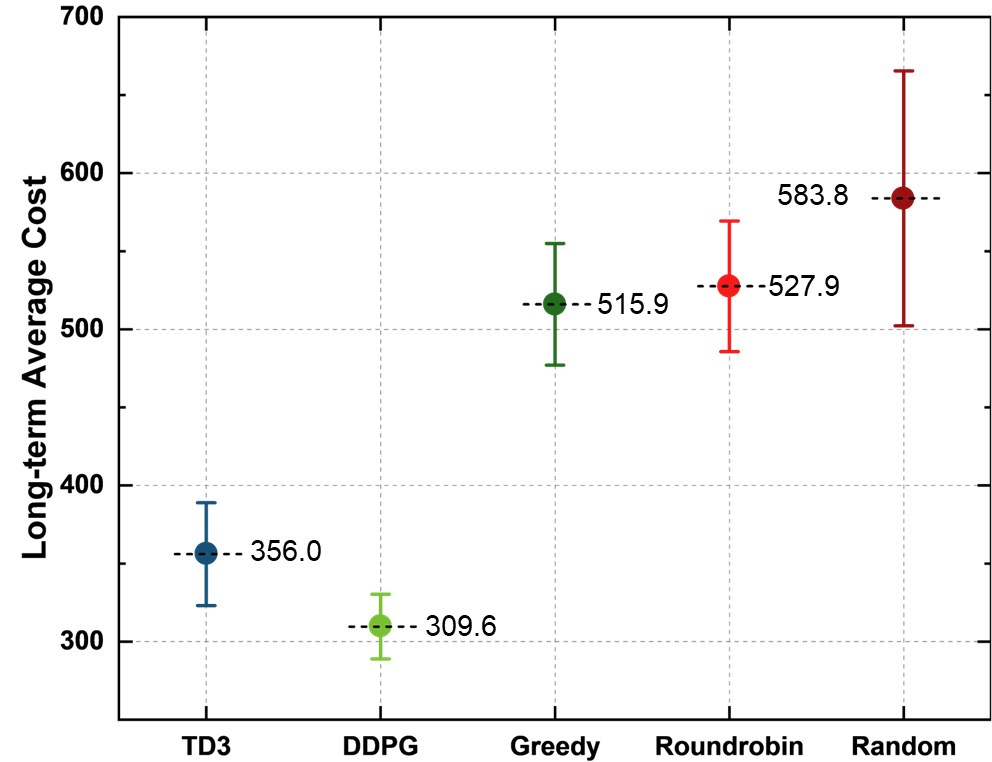}	
	\caption{The long-term average performance of DRL and benchmark algorithms over a system with $N = 8$ and $M = 6$. Values are means $\pm$ standard error of mean.}
	\label{fig:sim3}
\end{figure}

\subsection{Practical Example}

\begin{table}[t]
\centering
\footnotesize
\setlength\tabcolsep{3pt}
\caption{The Parameters of Pendulums}
\vspace{-0.3cm}
\begin{tabular}{m{2.8cm}m{0.5cm}m{0.5cm}m{0.5cm}m{0.5cm}m{0.5cm}m{0.5cm}m{0.5cm}m{0.5cm}}
\hline\hline
\textbf{Pendulum \#$i$} & 1 & 2 & 3  & 4 & 5  & 6 & 7 & 8\\ \hline\hline
Pole mass, $m_{p,i}$, [kg] & 1  & 1.2 & 1.4 & 1.6 & 1.8 & 2 & 2.2 & 2.4\\ 
Pole length, $l_{p,i}$, [m] & 1 & 1.1 & 1.2 &  1.3 & 1.4 & 1.5 & 1.6 & 1.7\\\hline\hline
\end{tabular}
\label{tab:PendulumParameters}
\vspace{-0.0cm}
\end{table}

We simulate a WNCS where 8 pendulums are spatially distributed and connected with a central controller via 6 frequency channels, as shown in Fig.~\ref{fig:PendulumCtrl}(a). The parameters of these pendulums, including pole mass $m_{p,i}$ and pole length $l_{p,i}$ for each pendulum $i$, are shown in Table~\ref{tab:PendulumParameters}. For pendulum $i$, the two-dimensional state vector $\mathbf{x}_{i,k}$, includes the pole angle state and the pole angular velocity state. The scaler control input $u_{i,k}$ is the applied torque. Each pendulum $i$ is linearized around the upright equilibrium using a small-angle approximation. The corresponding state transition matrix $\mathbf{A}_i$ and control input matrix $\mathbf{B}_i$ are given as \cite{10107379}
\begin{equation}
    \mathbf{A}_{i} =  \left[\begin{array}{cc}
         1 & T_s \\
         \frac{g}{l_{p,i}}T_s & 1
    \end{array} \right], \mathbf{B}_i =  \left[\begin{array}{c}
         T_s \\
         \frac{T_s}{m_{p,i}l_{p,i}^2}
    \end{array} \right],
\end{equation}
where $T_s = 0.01$ is the sampling time period; $g = 9.81$ is the gravitational acceleration. 
The measurement matrix $\mathbf{C}_i$ is equal to the identity matrix. The covariance matrices are $\mathbf{Q}^w_i = \mathbf{Q}^v_i = 0.001 \mathbf{I}$. Each pendulum is $2$-step controllable, i.e., $v_i=2$.

The packet success probabilities of each link, i.e., $\xi^{+}_{m,i}, \forall + \in \{s,c\}$, are simulated in the following. 
Let 
$\mathcal{F}_c \triangleq \{615\text{MHz}, 675\text{MHz}, 735\text{MHz}, 795\text{MHz}, 855\text{MHz}, 915\text{MHz}\}$ denote the available frequency bands. For links between each pendulum $i$ and the central controller, the free-space path loss model is $\bar{h}_i = A(\frac{3\times10^8}{4\pi f_id_i})^{d_e}$, where $A = 4.11$ denotes the antenna gain; $f_i \in \mathcal{F}_c$ denotes the carrier frequency; $d_i$ denotes the distance from the pendulum to the central controller and is randomly drawn from $[180 \text{m}, 200 \text{m}]$; $d_e = 2.2$ denotes the path loss exponent \cite{PathLossModel}. The average channel power gains for each link, $\mu^s_{m,i}$ and $\mu^c_{m,i}$, are randomly drawn from $[0.5, 1.5]$. Given the transmission power $P_{\text{tx}} = 23 \text{dBm}$ and the receiving noise power $\sigma^2 = -60 \text{dBm}$, the signal-to-noise ratio (SNR) of received packets in each link is obtained from $\gamma_{m,i}^{+} = \frac{\bar{h}_i\mu^{+}_{m,i} P_{\text{tx}}}{\sigma^2}$.

Given the packet length $l = 200$ (i.e., the number of symbols per packet), the number of data bits $b = 400$ in the packet, and the SNR $\gamma_{m,i}^{+}$ of the packet, we have the approximated decoding error probability of a packet as \cite{PangFBL}
\begin{equation}\label{BLER:NoHARQ}
\varepsilon\left(\gamma_{m,i}^{+}\right) \approx \mathcal{Q}\left(\frac{\mathcal{C}\left(\gamma_{m,i}^{+}\right)-\frac{b}{l}}{\sqrt{\frac{\mathcal{V}\left(\gamma_{m,i}^{+}\right)}{l}}}\right),
\end{equation}
where $\mathcal{C}(\gamma_{m,i}^{+})=\log_2{(1+\gamma_{m,i}^{+})}$ and $\mathcal{V}(\gamma_{m,i}^{+})=(1-(1+\gamma_{m,i}^{+})^{-2})(\log_2{e})^2$ are the Shannon capacity and the channel dispersion, respectively, and $\mathcal{Q}(x)=(\frac{1}{\sqrt{2\pi}})\int_{x}^{\infty}{e^{-\frac{t^2}{2}}\text{d}t}$ is the Gaussian Q-function. Then, packet success probabilities of each link can be expressed as $\xi^{+}_{m,i} = 1-\varepsilon\left(\gamma_{m,i}^{+}\right)$.
The resulting WNCS satisfies the stability condition established in Theorem \ref{theory:stability}.

The DDPG-based scheduler is trained with $\mathbf{S}^x_{i} = \left[\begin{array}{cc} 100 & 0 \\ 0 & 1\end{array} \right]$, $\mathbf{S}^u_{i} = 0.01$, and $\vartheta = 0.95$. After training, the control performance is shown in Fig.~\ref{fig:PendulumCtrl}(b). The pole angle of all pendulums is initialized at the down position. Results show that the pole angle of all pendulums can converge to zero, illustrating the effectiveness of the DDPG-based scheduler for stabilizing the pendulum system.

\begin{figure}[t]
    \centering
    \includegraphics[width=3.5in]{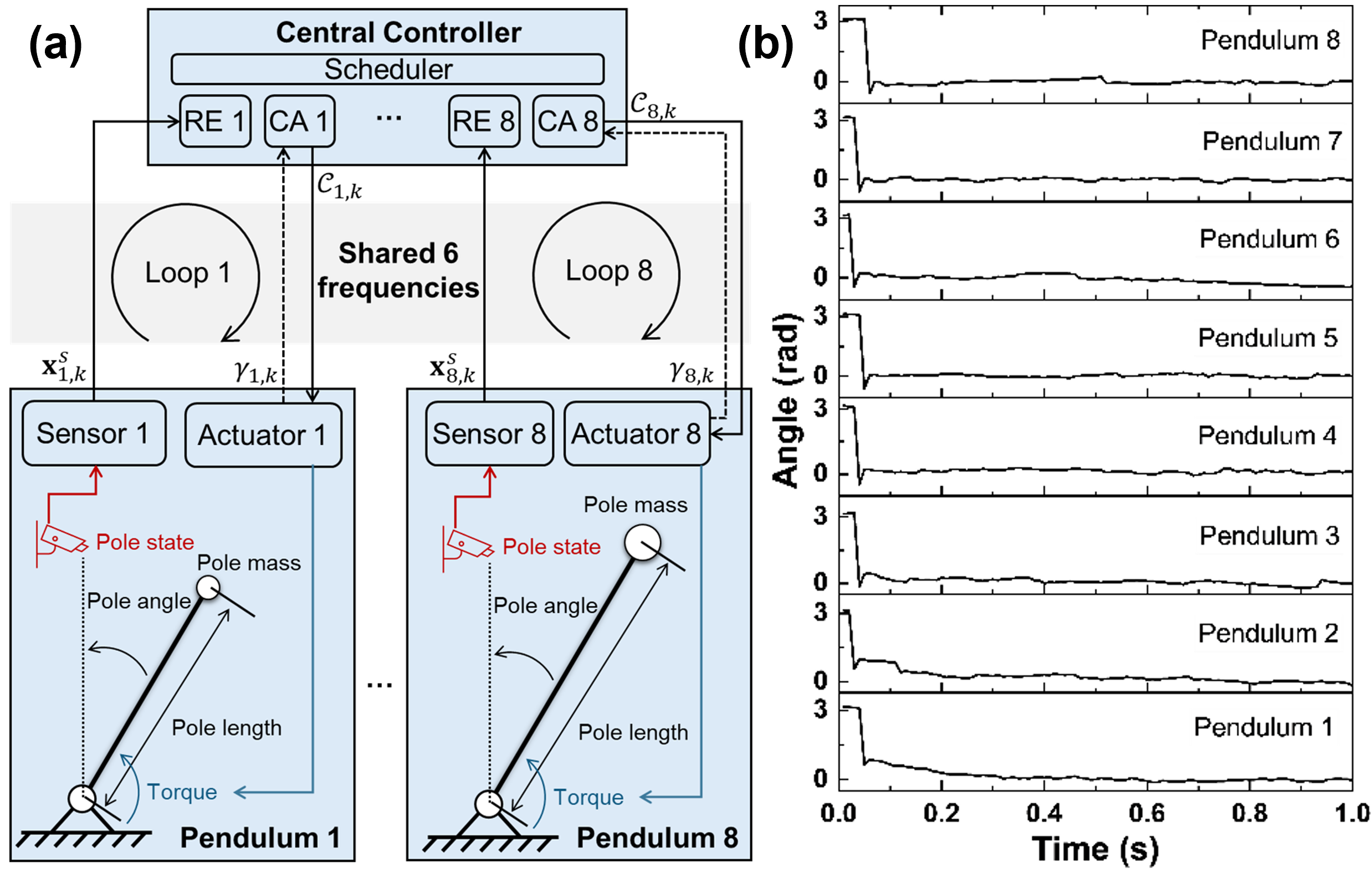}
    \vspace{-0.8cm}
    \caption{DDPG-based transmission scheduling for distributed control of the pendulum system: (a) Illustration of core components. (b) Testing results of the pole angle.} 
    \label{fig:PendulumCtrl}
    \vspace{0.0cm}
\end{figure}

\section{Conclusions}
We have investigated the transmission scheduling problem of the fully distributed WNCS.
A sufficient stability condition of the WNCS in terms of both the control and communication system parameters has been derived.
We have proposed advanced DRL algorithms to solve the scheduling problem, which perform much better than benchmark policies.
For future work, we will develop a distributed DRL approach for enhancing the scalability of the scheduling algorithm. Furthermore, we will consider co-design problems of the estimator, the controller, and the scheduler for distributed WNCSs with time-varying channel conditions.

\balance
\bibliographystyle{ieeetr}
%

\end{document}